\documentclass[11pt]{amsart}

\usepackage{url}
\usepackage{amssymb}
\usepackage{geometry}
 \usepackage{fullpage}
\usepackage{verbatim}
\usepackage{xcolor}
\newtheorem{theorem}{Theorem}
\newtheorem{remark}{Remark}
\newtheorem{lemma}{Lemma}
\newtheorem{corollary}{Corollary}
\newtheorem{proposition}{Proposition}

\newtheorem{hypothesis}{Hypothesis}
\newcommand{\tmop}[1]{\ensuremath{\operatorname{#1}}}

\def\pamod{\! \! \! \! \pmod}
\numberwithin{equation}{section}
\numberwithin{lemma}{section}
\numberwithin{proposition}{section}
\numberwithin{corollary}{section}
\numberwithin{theorem}{section}
\title[Superscars for arithmetic point scatterers II]{Superscars for arithmetic point scatterers II}
\author{P\"ar Kurlberg, Stephen Lester and Lior Rosenzweig}
\address{Department of Mathematics, KTH Royal Institute of Technology, SE-100 44 Stockholm, Sweden}
\email{kurlberg@kth.se}
\address{School of Mathematical Sciences, Queen Mary University, E1 4NS London, UK}
\email{s.lester@qmul.ac.uk}
\address{Mathematics Unit, Afeka Tel Aviv Academic College of Engineering, Mivtza Kadesh St., Tel-Aviv, Israel}
\email{LiorR@afeka.ac.il}
\date{\today}

\newcommand{\fixmehide}[1]{}
\newcommand{\fixmelater}[1]{}
\newcommand{\fixmedone}[1]{}
\newcommand{\fixmehidden}[1]{}
\newcommand{\Z}{{\mathbb Z}} 
 
\newcommand{\R}{{\mathbb R}} 
\newcommand{\C}{{\mathbb C}} 
\newcommand{\T}{{\mathbb T}} 

\begin{document}

\begin{abstract} 
  We consider momentum push-forwards of measures arising as
%  semi-classical limits of eigenfunctions of a point scatterer on the
  quantum limits (semi-classical measures) of eigenfunctions of a
  point scatterer on the standard flat torus $\T^2 = \R^2/\Z^{2}$.
  Given any probability measure arising by placing delta masses, with
  equal weights, on $\Z^2$-lattice points on circles and projecting to
  the unit circle, we show that the mass of certain subsequences of
  eigenfunctions, in momentum space, completely localizes on that
  measure and are completely delocalized in position (i.e.,
  concentration on Lagrangian states.)  We also show that the mass, in
  momentum, can fully localize on more exotic measures, e.g. singular
  continous ones with support on Cantor sets.  Further, we can give
  examples of quantum limits that are certain convex combinations of
  such measures, in particular showing that the set of quantum limits
  is richer than the ones arising only from weak limits of lattice
  points on circles.
The proofs exploit features of the
  half-dimensional sieve and behavior of multiplicative functions in
  short intervals, enabling precise control of the location
  of perturbed eigenvalues.
\end{abstract}

\maketitle

\section{Introduction}
\fixmehide{Minor style question: in the beginning I found it hard to find
  prop 2.1 (etc); we could have one sequence number for props, lemmas,
  thms etc. Any strong opinions here? (OTOH, it's sufficiently long
  that having section numbers show up is also good.)  I'm
  ambivalent... \textcolor{blue}{I don't have a preference, please feel free to change it as you wish} }

% One of
% the principal problems of quantum chaos is the following: Problem
% 3.1. Determine the set Q of ‘quantum limits’, i.e. weak* limit points
% of the sequence of invariant eigenfunction stats {ρk} or equivalently
% of the Wigner distributions {Wk}.

% Also note: Zelditch uses "quantum limit" to denote weak* limits of
% Wigner measures (or semi-classical measures.)  Good quote: he
% directly goes for large energy limit, convenient to avoid discussion
% of planck's constant etc.

Let $(M,g)$ be a smooth, compact Riemannian manifold with no boundary,
unit mass and let $\Delta_g$ denote the Laplace-Beltrami operator. Also,
let $\{ \phi_{\lambda} \}$ be an orthonormal basis of eigenfunctions
of $\Delta_g$ with eigenvalues
$0 \le \lambda_1 \le \lambda_2 \le \ldots$. For an observable
$f \in C^{\infty}(\mathbb S^* M)$, where $\mathbb S^* M$ denotes the
unit co-tangent bundle of $M$, let $\tmop{Op}(f)$ denote its
quantization, defined as a pseudo-differential operator
(cf. \cite{dimassi-sjostrand-book} for details.)  A central problem in
quantum chaos
(cf. \cite[Problem~3.1]{zelditch-recent-developments-qc}) is to
understand the set of possible quantum limits (sometimes called
semiclassical measures) describing the distribution of mass of the
eigenfunctions $\{ \phi_{\lambda} \}$ within $\mathbb S^* M$, in the
limit as the eigenvalue $\lambda$ tends to infinity.  A cornerstone result in
this direction is the quantum ergodicity theorem of Shnirelman
\cite{Shn}, Colin de Verdi\'ere \cite{YCD}, and Zelditch \cite{Z1}
which states that if the geodesic flow on $M$ is ergodic there exists
a density one subsequence of eigenfunctions $\{ \phi_{\lambda_j}\}$
such that
\[
\mu_{\phi_{\lambda_{j}}}(f) =\langle
\tmop{Op}(f) \phi_{\lambda_j}, \phi_{\lambda_j} \rangle \rightarrow \int_{\mathbb S^* M} f(x) d\mu_{L}(x)
\]
as $\lambda_j \rightarrow \infty$, where $d\mu_{L}$ is the normalized
Liouville measure on $\mathbb S^* M$.  
(Note that any quantum limit, by Egorov's theorem, is invariant under
the classical dynamics.)

While the quantum ergodicity theorem implies that the mass of almost
all eigenfunctions equidistributes in $\mathbb S^{*} M$ with respect
to $d\mu_{L}$, it does not rule out the existence of sparse
subsequences along which the mass of the eigenfunctions localizes.
Whether or not this happens crucially depends on the geometry of
$M$, cf. Section \ref{sec:discussion}. 
% For instance, if $M$ has constant negative curvature, Rudnick and
% Sarnak conjectured \cite{RS} that the mass of all eigenfunctions
% equidistributes. On the other hand, for a generic Bunimovich stadium,
% Hassell \cite{Hassell} has shown that there exists a subsequence of
% exceptional eigenstates where the mass localizes on the set of
% bouncing ball trajectories.  (Further examples exhibiting scarring
% will be discussed in Section~\ref{sec:discussion}.)
%(In the
%integrable case it is well known that mass can localize,
%e.g. ``whispering gallery modes''.)
\fixmelater{Maybe add that on irrational tori, say with eigenfunctions
  having even-even symmetry, we can get massive localization for ALL
  eigenfunctions.  BUT: point out that this setting has no QE.
  (Perhaps in
  discussion section?) \textcolor{blue}{It would be interesting to add, although we already have a very nice discussion so it might be too much}}

In this article we study quantum limits of ``point scatterers'' on
$M=\mathbb T^2=\mathbb R^2/2\pi \mathbb Z^2$.  These are singular
perturbations of the Laplacian on $M$, and were used by \v Seba
\cite{Seba} in order to study the transition between integrability and
chaos in quantum systems.\fixmehide{Maybe mention Kronig-Penney??}
The perturbation is quite weak and has essentially no effect on the
classical dynamics, yet the quantum dynamics ``feels'' the
effect of the scatterer, and an analog  of
the quantum ergodicity theorem  is known to hold
% by work of Kurlberg, Rudnick and Uebersch\"ar 
\cite{RU,KU} (namely, equidistribution holds for a full density
subset of the ``new'' eigenfunctions.)

The model also exhibits scarring along sparse subsequences of the new
eigenfunctions \cite{KR}. In particular there exist quantum limits
whose momentum push-forward, which can be viewed as probability
measures on the unit circle, is of the form
$c \mu_{\text{sing}} + (1-c) \mu_{\text{uniform}}$, for some
$c \in [1/2,1]$.  Here both $\mu_{\text{uniform}}$ and
$\mu_{\text{sing}}$ are normalized to have mass one, and
$\mu_{\text{sing}}$ can be taken to be a sum of delta measures giving
equal mass to the four points $\pm (1,0), \pm (0,1)$.
%and where $\mu_{\text{sing}}$ is microlocalized on two Lagrangian
%planes (hmm, maybe not - we only did it in 3d.)
We note that $\mu_{\text{uniform}}$ is the push-forward of the
Liouville measure and hence maximally delocalized, whereas
$\mu_{\text{sing}}$ is maximally 
localized since any quantum limits in this setting must be
invariant
under a certain eight fold symmetry (cf. \eqref{eq:sym8}).

Stronger localization, i.e., going strictly beyond $c=1/2$, is
particularly interesting given a number of ``half delocalization''
results for quantum limits for some other (strongly chaotic) systems,
namely quantized cat maps and geodesic flows on manifolds with
constant negative curvature.  For example, in the former case Faure
and Nonnenmacher showed
\cite{faure-nonnenmacher-maximal-cat-map-scarring} that if a quantum
limit $\nu$ is decomposed as
$\nu = \nu_{\text{pp}} + \nu_{\text{Liouville}} + \nu_{sc}$, with
$\nu_{\text{pp}}$ denoting the pure point part and $\nu_{sc}$ denoting
the singular continous part, then
$\nu_{\text{Liouville}}(\T^2) \ge \nu_{\text{pp}}(\T^2)$, and thus
$\nu_{\text{pp}}(\T^2) \le 1/2$. (We emphasize that $\T^{2}$ is the
full phase space in this setting.)

The aim of this paper is to exhibit essentially maximal
localization for a quantum ergodic 
system, namely arithmetic toral point scatterers.  In particular we
construct quantum limits (in momentum) corresponding to $c=1$ in the
above decomposition; other interesting examples include singular
continous measures with support, say, on Cantor sets.
This can be viewed as a step towards a ``measure classification'' for
quantum limits of quantum ergodic systems.\fixmelater{Maybe Reformulate/keep?
  Say that momentum push-forward avoids comparing different geometries
  of different systems? \textcolor{blue}{I like the way it's currently stated}}

\subsection{Description of the model}
\label{sec:description-model}
%
%
%  Understanding the different quantum limits was studied by Kurlberg and Rosenzweig \cite{KR}, who showed that there are sparse subsequences of eigenfunctions whose mass localizes within $\mathbb S^* M$. 
% The aim of this paper is to study the set of  quantum limits,
% continuing the investigation started by Kurlberg and Rosenzweig
% \cite{KR} (cf. Section~\ref{sec:discussion} for further details).  Our
% focus will be on the momentum push-forwards of these measures (for
% simplicity we will in what follows write ``in momentum
% space''.)\fixme{We kind of beat half delocalization! ``Maximul''
%   momentum scarring! (Note: we must have the obvious symmetries.)}

Let us now describe the basic properties of the point scatterer. This
is discussed in further detail in \cite{RU,RUspace,KU,KR,Seba,Shig}. 
To describe the quantum system associated with the point scatterer, consider
$
-\Delta|_{D_{x_0}}
$
where
\[
D_{x_0}=\{ f \in L^2(\mathbb T^2) : f(x) =0 \text{ in some neighborhood of } x_0 \}.
\]
By von Neumann's theory of self-adjoint extensions (see Appendix A of
\cite{RU}) there exists a one parameter family of self-adjoint
extension of $-\Delta|_{D_{x_0}}$ parameterized by a phase $\varphi
\in (- \pi, \pi]$. Moreover, for $\varphi  \neq \pi$ the eigenvalues
of these operators may be divided into two categories. The
\textit{old} eigenvalues which are eigenvalues of $-\Delta$, with
multiplicity decreased by one, along with \textit{new} eigenvalues
which are solutions to the spectral equation 
\begin{equation} \label{eq:spectral2}
\sum_{m \ge 1} r(m) \left( \frac{1}{m-\lambda}-\frac{m}{m^2+1} \right)=\tan(\varphi/2)\sum_{m\ge 1} \frac{r(m)}{m^2+1},
\end{equation}
where
\[
r(m)=\# \{ (a,b) \in \mathbb Z^2 : a^2+b^2=m\}.
\]
\fixmehide{Maybe add ref to Bogomolny et al? maybe not: they only
  consider the irrational aspect ratio case; in particular the assume
  poisson statistics for the unperturbed spectrum.}
We will refer to the case when $\varphi$ is fixed as $\lambda
\rightarrow \infty$ the \textit{weak coupling} quantization. 
In this regime work of Shigehara \cite{Shig} suggests that
the level spacing of the eigenvalues should have Poisson spacing
statistics and this is supported by work of Rudnick and Uebersch\"ar
\cite{RUspace} along with Freiberg, Kurlberg and Rosenzweig
\cite{FKR}. In hope of exhibiting wave chaos Shigehara proposes the
following \textit{strong coupling} quantization 
\begin{equation} \label{eq:spectral}
\sum_{|m-\lambda|\le \lambda^{1/2}} r(m) \left( \frac{1}{m-\lambda}-\frac{m}{m^2+1} \right)=\frac{1}{\alpha},
\end{equation}
where $\alpha \in \mathbb R$ is called the physical coupling constant
and reflects the strength of the scatterer. The strong coupling
quantization restricts the spectral equation to the physically
relevant energy levels. Notably, this forces a re-normalization of
\eqref{eq:spectral2} 
\[
\tan(\varphi/2)\sum_{m \ge 1} \frac{r(m)}{m^2+1} \sim -\pi \log \lambda
\]
so that $\varphi$ depends on $\lambda$ in this case (see \cite{Ue}
equation (3.14)).  We note that the weak coupling quantization corresponds to
a fixed self adjoint extension, whereas the strong coupling quantization can
be viewed as an energy dependent, albeit very slowly varying, family
of  self
adjoint extensions.

% construction the resulting new eigenvalues are attracted
% where $\varphi \in (\pi,\pi]$ parameterizes the family self adjoint
% extensions of $-\Delta$ which act on functions which vanish near
% $x_0$. We are mainly interested in the strong\fixmehidden{Can we say
%   something in the weak case?  Probably everything works there too? \textbf{The analog of Theorem \ref{thm:scar} seems to follow. But, I'm not sure this is the case for the convex combinations result. A main input into this result is the existence of a subsequence of eigenvalues such that $\delta_{\lambda} \asymp 1$ (where $\delta_{\lambda}$ is the distance from $\lambda$ to the nearest integer representable as a sum of two squares). This step does not carry over to the weak coupling case using our argument (in that case the distance is much smaller $\delta_{\lambda} \asymp 1/\log \lambda$).  }}
% coupling limit where 
% $\varphi$ is allowed to depend on $\lambda$. In this setting it is
% known that by using results from the Gauss circle problem the spectral
% equation may truncated and the new eigenvalues of
% $-\Delta+\delta_{x_0}$ are given by solutions $\lambda$ to
%  \begin{equation} \label{eq:spectral}
%  \sum_{|m-n| \le \lambda^{1/3}} r(m) \left( \frac{1}{m-\lambda}-\frac{m}{m^2+1} \right)=0,
%  \end{equation}
%  where $n$ is the largest integer that is a sum of two squares and is less than $\lambda$. 

 From the spectral equation it follows that new eigenvalues interlace
with integers which are representable as the sum of two integer squares.
We denote these eigenvalues as follows
\[
0< \lambda_0 < 1 < \lambda_1 < 2 < \lambda_2 <4 < \lambda_4 < 5< \lambda_5 < \cdots
\]
and write $\Lambda_{new}$ for the set of all such eigenvalues. Also, given $n=a^2+b^2$ let 
$n^+$ denote the smallest integer greater than $n$ which is also a sum of two squares. Let
\begin{equation} \label{eq:deltadef}
\delta_n=\lambda_n-n>0,
\end{equation}
(which should not be confused with the Dirac delta function).
In addition given $\lambda \in \Lambda_{new}$ the associated Green's
function is given by 
\begin{equation} \label{eq:glambda}
G_{\lambda}(x)=-\frac{1}{4 \pi^2} \sum_{ \xi \in \mathbb Z^2} \frac{ \exp(-i \xi \cdot x_0)}{|\xi|^2- \lambda} e^{i \xi \cdot x}, \qquad g_{\lambda}(x)=\frac{1}{\lVert G_{\lambda} \rVert_2} G_{\lambda}(x),
\end{equation}
(see equation (5.2) of \cite{RU}).  Since the torus is homogeneous we
may without loss of generality assume that $x_{0}=0$.

\subsection{Results}
\label{sec:results}
Our first main result shows that along a sparse, yet relatively large,
subsequence of new eigenvalues $\{\lambda_j\}$ that the mass of
$g_{\lambda_j}$ in momentum space localizes on measures arising from $\Z^2$-lattice points on circles, projected to the unit
circle. To describe these measures in more
detail, consider an integer $n=a^2+b^2$, with $a,b\in \Z$, and the
following probability 
measure on the unit circle $S^1 \subset \C$
\[
\mu_n =\frac{1}{r(n)} \sum_{a^2+b^2=n} \delta_{(a+ib)/|a+ib|}.
\]
Following Kurlberg and Wigman \cite{KW} we call a measure
$\mu_{\infty}$ \textit{attainable} if it is a weak limit point of the
set $\{ \mu_n \}_{n=a^2+b^2}$.  Any such measure is 
invariant under rotation by $\pi/2$, as well as under reflection in
the $x$-axis; for convenience let
\begin{equation}
  \label{eq:sym8}
\text{Sym}_{8} := 
\left\{
\left\langle
\begin{pmatrix} 0 & -1 \\ 1 & 0  \end{pmatrix},
\begin{pmatrix} -1 & 0 \\ 0 & 1  \end{pmatrix}
\right\rangle \right\} 
\subset  GL_{2}(\Z)
\end{equation}
denote the group
generated by these transformations.

\begin{theorem} \label{thm:scar} 
% \fixme{I think we need to ``sell''
%     the theorem(s) better. Is there anything conjectured about quantum
%     limits? PK: Hopefully now better. TY}
\fixmehide{MAybe add: without perturbation,can break 8fold symmetry (cf. Dima)}
Let $m_0=a^2+b^2 \in \mathbb N$ be odd\footnote{As far as possible
  quantum limits go, $m_{0}$ being odd is not a restriction as any
  $\mu_{n}$ for $n$ even can be approximated by $\mu_{m_0}$ for
  $m_{0}$ odd.}.  In each of the weak and strong
coupling quantizations there exists a subset of eigenvalues $\mathcal
E_{m_0} \subset \Lambda_{new}$ 
with
\[
\frac{
\# \{ \lambda \le X : \lambda \in \mathcal E_{m_0} \}}{\# \{ \lambda \le X : \lambda \in \Lambda_{\text{new}} \} } \gg  \frac{1 }{(\log X)^{1+o(1)}}
\]
such that for any pure momentum observable $f \in C^{\infty} (S^1)
\subset C^{\infty}(\mathbb 
S^{*}(\mathbb T^2))$
% \change{ changed here $M$ to $\mathbb T^2$. also
%   called f pure momentum observable here (and similarly later in cor
%   1.1 and thm 1.2)} 
\[
\langle \tmop{Op}(f) g_{\lambda}, g_{\lambda} \rangle \xrightarrow{ \substack{\lambda \rightarrow \infty \\ \lambda \in \mathcal E_{m_0}}} \frac{1}{r(m_0)} \sum_{a^2+b^2=m_0} f\left(\frac{a+ib}{|a+ib|} \right).
\]
\end{theorem}

%\begin{remark} 
We note that the quantization of our observables is as explicitly
given in \eqref{eq:quantization}, which follows the approach of
\cite{KU}.
%\end{remark}

% \begin{remark}
%   In our proof we construct a set $S \subset \mathbb N$ such that for each $m_0=a^2+b^2$ we have $\mathcal E_{m_0} \subset \{ \lambda_n \in \Lambda_{new} : n \in m_0 S \} $.
% \end{remark}

Hence, in momentum space the mass of
$g_{\lambda}$ completely localizes on the measure  
$
\mu_{m_0}$. For any attainable measure $\mu_{\infty}$ there exists
$\{m_{0, \ell} \}_{\ell}$ such that $\mu_{0,\ell}$ weakly converges to
$\mu_{\infty}$. This implies the following
corollary. 

\begin{corollary}Let $\mu_{\infty}$ be an attainable measure. Then there
  exists
%   \fixme{Why not just take one subset of the new spectrum, and
%     let $\lambda$ grow inside this set?  I.e., pass to sub-subsequence. I agree}
  $\{\lambda_j\}_{j} \subset \Lambda_{\text{new}}$ such
  that for any pure momentum observable $f \in
  C^{\infty}(S^1) $
\[
 \langle \tmop{Op}(f) g_{\lambda_j}, g_{\lambda_j} \rangle \xrightarrow{j \rightarrow \infty} \int_{S^1} f d\mu_{\infty}.
\]
\end{corollary}

% To specify the next result we need to introduce a subclass of limiting measures $\mathcal A_{\mathcal I,k}(S^1)$ which are limits of measures
% \[
% \frac{1}{r(n)} \sum_{a^2+b^2=n} \delta_{(a+ib)/|a+ib|}
% \]
% as $n \rightarrow \infty$
% for which $n$ has prime divisors $p_1, \ldots p_k$ with $\theta_{p_1} \in I_1$, \ldots ,$\theta_{p_k} \in I_k$.

We note that the set of attainable measures is much smaller than the
set of probabality measures on $S^{1}$ that are
$\text{Sym}_{8}$-invariant, in particular the set of attainable
measures is {\em not convex} (cf. \cite[Section~3.2]{KW}.)
In our next result we show that in the strong coupling quantization
there is a subsequence of new eigenvalues along which the entire mass
of $g_{\lambda}$ localizes on certain convex combination of two measures
arising from lattice points on the circle.  In particular, the set of
quantum limits, in momentum space, is {\em strictly richer} than the
set of attainable measures.
% To state the this we introduce the following notation. For $\lambda \in \Lambda_{new}$ let $n=n_{\lambda}$ be the greatest integer less than $\lambda$. Also let $\delta_n=\lambda-n$. 

%  \fixme{Question: can we use Maynard type
%   techniques to put many of ``our'' types of integers around some new
%   eigenvalue?  If that new eigenvalue separates ok from the old ones,
%   we'd get a convex combination of many measures.  Worth doing!? \textbf{Would be nice to do this! The main problem is that Maynard's work does not give the correct order of magnitude for the lower bound. This causes serious problems for convex combinations, since our estimate for the tail end of the spectral equation isn't strong enough, see Lemma \ref{lem:cheby}. Note that in our earlier set-up where we had $n \equiv 1 \pmod Q$, with $Q$ a product of a lot of small primes $1 \pmod 4$, we could get away with much weaker estimates.}}

\begin{theorem} \label{thm:scar2}  
Let $m_0,m_1$ be odd integers
which are each 
representable as a sum of two squares. 
Then in the strong coupling quantization there exists a subsequence of
eigenvalues $ \mathcal E_{m_0,m_1} \subset \Lambda_{\text{new}}$ such that for each $\lambda \in  \mathcal E_{m_0,m_1}$
there is an integer $\ell_{\lambda}$ with $r(\ell_{\lambda}) \neq 0$ and
$r(\ell_{\lambda}) \ll 1$ 
such that for  pure momentum observables $f \in C^{\infty}(S^1)$ 
\begin{equation} \label{eq:mainresult}
\begin{split}
\langle \tmop{Op}(f) g_{\lambda}, g_{\lambda} \rangle =&  c_{\lambda} \cdot \frac{1}{r(m_0)} \sum_{a^2+b^2=m_0} f\left( \frac{a+ib}{|a+ib|}\right) \\&  + (1-c_{\lambda}) \cdot \frac{1}{r(m_1\ell_{\lambda})} \sum_{a^2+b^2=m_1 \ell_{\lambda}} f\left( \frac{a+ib}{|a+ib|}\right)+O\left(\frac{1}{(\log \log \lambda)^{1/11}} \right),
\end{split}
\end{equation}
where
\[
c_{\lambda}=\frac{1}{1+r(m_0)/r(m_1\ell_{\lambda})}.
\]
Additionally, 
\[
\frac{
\# \{ \lambda \le X  : \lambda \in \mathcal E_{m_0,m_1} \}}{\# \{ \lambda \le X : \lambda \in  \Lambda_{\text{new}} \} } \gg  \frac{1 }{(\log X)^{2+o(1)}}.
\]
\end{theorem}
% \noindent

Note that since $\sum_{p|\ell_{\lambda}} 1 \ll 1$, the measure
  $\mu_{m_{1} \ell_{\lambda}}$ can be viewed as a fairly small
  perturbation of $\mu_{m_{1}}$.

\begin{remark}
  By removing a further ``thin'' set of eigenvalues (with spectral
  counting function of size $O(x^{1-\epsilon})$ for $\epsilon>0$, we
  can construct quantum limits that are flat in position
  (for details, cf. \cite[Remark~4]{KR}), in 
  addition to the momentum push-forward properties given in Theorems
  \ref{thm:scar} and \ref{thm:scar2}.  In particular, we can construct
  quantum limits that are completely localized on the superposition of
  two Lagrangian states --- essentially two plane waves, one in the
  horizontal and one in the vertical direction.  This phenomena is
  sometimes called {\bf super scarring}
  (cf. \cite{bogomolny-schmit-superscars04,KR}.)
\end{remark}

% Beyond controlling the number of prime factors of $\ell_{n_j}$ we also know that each such prime factor is $\ge \lambda_{n_j}^{\theta}$ for some fixed $\theta >0$, so in particular $(\ell_{n_j}, m_1)=1$. 
Further, assuming a plausible conjecture on the distribution of the prime numbers, we show that given $m_0,m_1$ as in Theorem \ref{thm:scar2} the quantum limit of
  $\langle \tmop{Op}(f) g_{\lambda}, g_{\lambda} \rangle$
  is a convex combination
  of $\mu_{m_0}$ and $\mu_{m_1}$. From this we are able to conclude that \textit{every} $\text{Sym}_{8}$-invariant measure arises as a quantum limit.
   The conjecture on the distribution of primes concerns obtaining a lower bound on the number solutions $(u,v)$ in almost primes to the Diophantine equation
\[
aX-bY=4
\]
 where $v=p_1p_2$, $u=p_3$ with $p_j$ a prime satisfying $p_j=a_j^2+b_j^2$ and $b_j=o(a_j)$ for $j=1,2,3$. The precise formulation of this conjecture, which we call Hypothesis \ref{hyp:twinprimes} is given in Section \ref{sec:twinprimes}.
%   The
%   following conjecture thus seems natural.
%   \change{formal conjecture made}
  \begin{theorem} \label{thm:allmeasures}
    Assume Hypothesis \ref{hyp:twinprimes}.
%   Moreover, from this it is not hard
%   to deduce that for
   Let
  $\mu_{\infty_0}, \mu_{\infty_1}$ be attainable measures and $0\le c \le 1$. Then in the strong coupling quantization there exists $\{ \lambda_j \}_{j} \subset \Lambda_{new}$ such that for any $f \in C^{\infty}(S^1)$
  \[
\langle \tmop{Op}(f) g_{\lambda_j}, g_{\lambda_j} \rangle 
\xrightarrow{ j \rightarrow \infty}  c\int_{S^1} f d\mu_{\infty_0} +(1-c) \int_{S^1} f d\mu_{\infty_1}.
  \]
  In particular, {\bf all} $\text{Sym}_{8}$-invariant probability
  measures on $S^{1}$ arise as quantum
  limits in momentum space.
  \end{theorem}

We finally remark that the proof of Theorem 1.2 easily (and
unconditionally) also gives that any $\text{Sym}_{8}$-invariant probability
measure $\mu$ on $S^{1}$ is a quantum limit of Greens function in the
following sense: given $\mu$, there exist a sequence of positive reals
$\lambda_1' < \lambda_2'< \cdots$, disjoint from the set of unperturbed
eigenvalues, so that
$ \lim_{i \to \infty} \langle \tmop{Op}(f)
g_{\lambda_i'}, g_{\lambda_i'} \rangle  = \mu.
$

\subsection{Discussion}
\label{sec:discussion}
\fixmehide{PK: maybe mention stuff in paper
  with  henrik (but this is not a QE system - maybe
  not?).}
For integrable systems it is often straightforward to construct
non-uniform quantum limits, e.g. ``whispering gallery modes'' for the
geodesic flow in the unit ball, and for linear flows on $\T^2$,
Lagrangian states with maximal localization (i.e., a single plane
wave) are easily constructed.  We note that strong localization in
position for quantum limits on $\T^2$ was ruled out by Jakobson
\cite{jakobson-quantum-limits-tori} --- in position, any quantum limit
is given by trigonometric polynomials whose frequencies lie on at most
two circles (hence absolutely continuous with respect to Lebesgue
measure.)  Further, for the sphere, Jakobson and Zelditch in fact
obtained a full classification --- {\em any} flow
invariant measure on $S^{*}(S^{2})$ is a quantum limit
\cite{jakobson-zelditch-anything-on-sphere}.
% \change{dima-steve paper mentioned}

The quantum ergodicity theorem holds in great generality as long as
the key assumption of ergodic classical dynamics holds, but the
existence of exceptional subsequence of nonuniform quantum limits
(``scarring'') is subtle.  For classical systems given by the geodesic
flow on compact negatively curved manifolds, the celebrated Quantum
Unique Ergodicity (QUE) conjecture \cite{RS} by Rudnick and Sarnak
asserts that the only possible quantum limit is the Liouville
measure. Known results for QUE include Lindenstrauss' breakthrough
\cite{lindenstrauss-arithmetic-que} for Hecke eigenfunctions on
arithmetic modular surfaces, together with Soundararajan ruling out
``escape of mass'' in the non-compact case \cite{sound-que}.
On the other hand, for a generic Bunimovich stadium (with strongly chaotic
classical dynamics), Hassell \cite{Hassell} has shown that there exists a
subsequence of exceptional eigenstates where the mass localizes on 
sets of bouncing ball trajectories.
% \change{hassel mentioned here}

For quantized cat maps, again for Hecke eigenfunctions, QUE is know to
hold \cite{cat1}.  However, unlike for arithmetic modular surfaces,
where Hecke desymmetrization is believed to be unnecessary, it is
essential for quantum cat maps. Namely, Faure, Nonnenmacher and de
Bi{\`e}vre \cite{faure-nonnenmacher-bievre-scarred-cat} constructed,
in the presence of extreme spectral multiplicities and no Hecke
desymmetrization, quantum limits of the form
$\nu = \frac{1}{2} \nu_{\text{pp}} + \frac{1}{2}
\nu_{\text{Liouville}}$;
in \cite{faure-nonnenmacher-maximal-cat-map-scarring} this was shown
to be sharp in the sense that the Liouville component always carries
at least as much mass as the pure point one.  (We note that, on
assuming very weak bounds on spectral multiplicities, Bourgain showed
\cite{bourgain-cat-maps} that scarring does not occur.)  For higher
dimensional analogs of quantum cat maps, Kelmer has for certain maps
shown \cite{kelmer-que-and-superscars-symplectic} ``super scarring'',
even after Hecke desymmetrization, on invariant rational isotropic
subspaces.  Further, these type of scars persist on adding certain
perturbations that destroy the spectral multiplicities
\cite{kelmer-scarring-mult-one}. Other models where scarring is known to exist include
toral point scatterers with irrational aspect ratios
\cite{superscars,keating-marklof-winn-seba-scars,BKW03}
and quantum star graphs \cite{keating-etal-no-qe-star-graphs}, though
neither  model is quantum ergodic
\cite{superscars,keating-etal-no-qe-star-graphs}.

Classifying the set of possible quantum limits, in particular
for Quantum Ergodic settings, is an interesting question. Here
Anantharaman proved very strong results for geodesic flows on
negatively curved manifolds \cite{anantharaman-entropy-bounds}: any
quantum limit has positive Kolmogorov-Sinai (KS) entropy with respect
to the dynamics of the geodesic flow. In particular, this rules out
localization on a finite number of closed geodesics (for compact
arithmetic surfaces this was already known due to Rudnick and Sarnak
\cite{RS}.)  Moreover, in the case of constant negative curvature,
Anantharaman and Nonnenmacher showed
\cite{anantharaman-nonnenmacher-half-delocalization} that the
KS-entropy is at least half of the maximum possible.  The measure of
maximum entroy is given by the Liouville measure, and thus
``eigenfunctions are at least half delocalized''.
Dyatlov and Jin
\cite{dyatlov-jin-hyperbolic-semiclassical-measures-full-support}
consequently showed that any quantum limit must have {\em full} support in
$S^{*}(M)$, for compact hyperbolic surfaces $M$  with constant
negative curvature; together with Nonnenmacher this was recently
strengthened \cite{dyatlov-jin-nonnenmacher-control} to the include the
case of surfaces with variable 
negative curvature.
% \change{dyatlov-jin and nonnenmacher papers mentioned}

\fixmelater{Maybe add here further comparison/discssion with superscars?}

\subsection{Outline of the proofs}
\label{sec:outline-argument}
\fixmehide{pk: add a discussion of sieve methods } 
  Our arguments use the multiplicative structure of the integers to
  create an imbalance in the spectral equation \eqref{eq:spectral}
  along a zero density, yet relatively large subsequence of new
  eigenvalues. Through exploiting this imbalance we control the
  location of the new eigenvalues in our subsequence and show that
  they lie close to integers which are sums of two squares. This
  greatly amplifies the amount of mass of the corresponding
  eigenfunctions in momentum space which lies on the terms which
  correspond to these integers, so much so that the contribution of
  the remaining terms is negligible. Consequently, the mass completely
  localizes on a convex combination of two measures and moreover our
  construction allows us to completely control the first measure. 
  
In Section \ref{sec:sieve} we use sieve methods to produce
  integers $n=p_1p_2$
where $p_j$, $j=1,2$, is a prime with
  $p_j=a^2+b^2=(a+ib)(a-ib)$,  $0< b \le a$, with $0\le \arctan(b/a)
  \le \varepsilon$, where $\varepsilon$ is a small parameter, such
  that $Q_0p_1p_2+4$ is also a sum of two squares, $Q_1|Q_0p_1p_2+4$
  and $(Q_0p_1p_2+4)/Q_1$ has a bounded number of prime factors, where
  $Q_0,Q_1$ are large integers whose purpose we will describe
  later.  
    In particular, we exploit special features of the half dimensional
    sieve using an ingenious observation of Huxley and
    Iwaniec \cite{huxley-iwaniec}.
    Further, in order to find suitable Gaussian primes in narrow sectors we use
    a classical result of Hecke together with non-trivial bounds on
    exponential sums over finite fields to control sums of integral
    lattice points in narrow sectors with norms lying in arithmetic
    progressions to large moduli. 
  
 The subsequence of almost primes $\{n_{\ell}\}$ constructed as described above creates the imbalance in the spectral equation \eqref{eq:spectral} by boosting the contribution of the terms $m=Q_0n_{\ell}, Q_0n_{\ell}+4$. The next step in our argument is to show that this imbalance typically overwhelms the contribution of the remaining terms. To do this, we first show in Section \ref{sec:truncate} that for all new eigenvalues lying outside a small exceptional set the spectral equation \eqref{eq:spectral} can be effectively truncated to integers $m$ with essentially $|m-\lambda| \ll (\log \lambda)^{10}$.
  This is done by controlling sums of $r(n)$ over short intervals and uses a second moment estimate of the Dedekind zeta-function $\zeta_{\mathbb Q(i)}$.
In Section \ref{sec:normal} we apply this result to new eigenvalues which lie between $Q_0n_{\ell}$ and $Q_0n_{\ell}+4$ and show that for almost all such new eigenvalues the remaining terms in the spectral sum (i.e. $|m-\lambda| \ll (\log \lambda)^{10}, m \neq Q_0n_{\ell}, Q_0n_{\ell}+4)$ is relatively small, provided that we take $Q_0,Q_1$ sufficiently large thereby boosting the contribution of the closest two terms. This is accomplished by using bounds for sums of multiplicative functions over polynomials due to Henriot \cite{H}. Crucially, we need good estimates for these sums in terms of the discriminant of the polynomials. 

Finally, to get complete control on the first measure in Theorem \ref{thm:scar2} we choose $Q_0$ so that it is the product of a given fixed integer $m_0$ and large primes $p_{k}=a^2+b^2$ with $0 \le \arctan(b_k/a_k) \le p_k^{-1/10}$ so that the probability measure on $S^1$ associated with $Q_0n_{\ell}$ weakly converges to the measure associated with $m_0$ as $\ell \rightarrow\infty$.  This last construction uses work of Kubilius \cite{Kub} on Gaussian primes in narrow sectors.

\subsection{Notation} We write $f(x) \ll g(x)$ provided that
$f(x)=O(g(x))$. Additionally, if for all $x$ under consideration
$|f(x)| \ge c g(x)$ we write $f(x) \gg g(x)$. If we have both $f(x)
\ll g(x)$ and $f(x) \gg g(x)$ we write $f(x) \asymp g(x)$.  For 
some additional notation related to sieves, see Section~\ref{sec:notation-sieve}.

% \subsection{Acknowledgements}
% \label{sec:acknowledgements}

\subsection{Acknowledgments}
P.K. was partially supported by the Swedish Research Council
(2016-03701).
We would also like to thank Dimitris Koukouloupolos,
St\'ephane Nonnenmacher,
Maksym Radziwi\l \l,  Ze{\'e}v Rudnick, and  Steve Zelditch for very
helpful discussions and suggestions.
% \change{thanking more people} 
%

\section{Sieve estimates} \label{sec:sieve}  Let $B_0$ be a
sufficiently large integer, and given $\varepsilon > 0$ let  
\begin{equation} \label{eq:Pepsidef} 
\begin{split}
\mathcal P_{\varepsilon}=&\{ p \ge (\log x)^{B_0}  : p=a^2+b^2 \,
\text{ and } \, 0 < \arctan(b/a) \le \varepsilon
\}, \\
\mathcal P_{\varepsilon}'=&\{ p \in \mathcal P_{\varepsilon} : p \le x^{1/9}\} .
\end{split}
\end{equation}
Throughout we assume that $\varepsilon \ge 1/(\log \log x)^{1/2}$ is sufficiently small. 
Also given $f,g : \mathbb N \rightarrow \mathbb C$ we define the
Dirichlet convolution of $f$ and $g$ by
\[
(f \ast g)(n)=\sum_{ab=n} f(a)g(b).
\]
Also, let $Q_0,Q_1 \le (\log x)^{1/10}$ be odd co-prime integers whose
prime factors are all $\equiv 1 \pmod 4$. Moreover we assume that
$Q_0=f_0^2 e_0 r_0^{a_0}, Q_1=f_1^2e_1 r_1^{a_1}$ where $e_0,e_1$ are
square-free, $f_0,f_1 \ll 1$ and $r_0,r_1$ are primes congruent to
$1 \pmod 4$. Throughout, the arithmetic function $b(n)$ is the indicator
function of the set of integers which are representable as a sum of
two squares.
    Also, for $\mathcal S \subset \mathbb N$ we define
\[
1_{ \mathcal S}(n)=\begin{cases}
1 & \text{ if } n \in  \mathcal S, \\
0 & \text{ otherwise.}
\end{cases}
\]
and let 
$\varphi(n)=\#\{ m < n : (m,n)=1\}$.
% \change{moved definition of
%   $\varphi$ to first 
%   usage. NOTE: we also use $\varphi$ for quantization parameter
%   earlier, but I think the clash is not too bad.}

% {\tt We can take $\varepsilon$ somewhat smaller, even $\varepsilon=1/x^{\delta}$ should be possible, but would require much more work.}
\begin{proposition} \label{prop:sieve} 
Let $\eta>0$ be sufficiently small and let $y=x^{\eta}$. Suppose $y>Q_0Q_1$. Then
\[
\sum_{\substack{n \le x \\ Q_1 | Q_0 n+4 \\ (\frac{Q_0n+4}{Q_1},\prod_{p \le y} p)=1}} (1_{ \mathcal P_{\varepsilon}} \ast 1_{ \mathcal P_{\varepsilon}'})(n) b(Q_0n+4) \ge \frac{C \varepsilon^{2} Q_0}{ \eta^{1/2} \varphi(Q_0)} \cdot \frac{x \log \log x}{ \varphi(Q_1) (\log x)^2},
\]
for some absolute constant $C>0$.
\end{proposition}
This proposition builds on a result of Friedlander and
Iwaniec
\cite[Ch.~4]{FI}. The main novelty here is that we capture almost primes
$n =p_1p_2$ such that each prime factor $p=a^2+b^2$, with
$0 \le b \le a$, has the property that $a+ib$ lies within a certain
small sector.

We also will require the following result.
\begin{proposition} \label{prop:sieve2}
There exists an absolute constant $C>0$
such that
\[
\sum_{\substack{n \le x \\ Q_1 | Q_0 n+4 }} (1_{ \mathcal P_{\varepsilon}} \ast 1_{ \mathcal P_{\varepsilon}'})(n) b(Q_0n+4) \ge C \varepsilon^2 \frac{x \log \log x}{ \varphi(Q_1) (\log x)^{3/2}}.
\]
\end{proposition}

Since Proposition \ref{prop:sieve2} follows from a similar, yet simpler argument than the one used to prove Proposition \ref{prop:sieve} we will omit its proof. The rest of this section will be devoted to proving Proposition \ref{prop:sieve}.

% \marginpar{This proposition should be able to be greatly improved. It should be able to do this for primes, not almost primes, also I believe one can get the correct order of magnitude here and allow $\varepsilon \rightarrow 0$ in terms of $x$, probably a negative power of $x$.}

% {\tt The strange looking condition $\eta^{1/(20 \eta^{1/2})} <\varepsilon$ arises from \eqref{eq:sievelast}, also need $\delta> \eta$ \textcolor{blue}{note the slight change in exponent of $\eta$ here, from the previous draft}}

% \begin{remark} Taking $\eta=1/(\log x)$  gives\fixme{$\eta=1/(\log x)$
%   seems to give $y = O(1)$, this seems like potential trouble? \textbf{ This is confusing, but I want to also include the case $P(y)=O(1)$, which is what is relevant for just one measure}}
% \begin{equation} \label{eq:sieve1}
% \sum_{n} (1 \ast 1)(n) b(m_0n+4) \ge C \varepsilon^3 \frac{x \log \log x}{(\log x)^{3/2} },
% \end{equation}
% which we will also use later.
% \end{remark}

% \begin{remark} It should be possible to prove a version of Proposition \ref{prop:sieve} where $(1_{\mathcal P}\ast 1_{\mathcal P})(n)$ is replaced by $1_{\mathcal P}(n)$. Furthermore, it should also be possible to give a lower bound of the correct order of magnitude in terms of $\varepsilon$ and to allow much smaller $\varepsilon$. However, since Proposition \ref{prop:sieve} suffices for our purposes we leave these problems open. For instance, significant progress in these directions has been made by Coleman ... and Coleman and Swallow.
% \end{remark}

\subsection{The Rosser-Iwaniec Sieve}

Let us first introduce the Rosser-Iwaniec $\beta$-sieve and the classical sieve terminology.
We start with a sequence of $\mathcal A=\{ a_n \}$ of non-negative real numbers, a set of primes $\mathcal P$ and a parameter $z$. Define
\[
P(z)=\prod_{\substack{p \in \mathcal P \\ p < z}} p.
\]
Our goal is to obtain an estimate for the sieved set
\[
\mathcal S(\mathcal A, \mathcal P, z):=\sum_{\substack{n \le x \\ (n, P(z))=1}} a_n.
\]

This will be accomplished through calculating, for square free
$d \in \mathbb{N}$,
\begin{equation}
  \label{eq:Ad1}
A_d(x):=\sum_{\substack{n \le x \\ n \equiv 0 \pmod d}} a_n.
\end{equation}
We now make the hypothesis that our estimate for $A_d(x)$ will be of the form
\begin{equation}
  \label{eq:Ad2}
A_d(x)=g(d)X+r_d  
\end{equation}
where $g(d)$ is a multiplicative function with $0 \le g(p)<1$. The
number $r_d$ should be thought of as a remainder term, so $X$ is an
approximation to $A_1(x)$, and the function $g(d)$ can be interpreted as a
density. 

Let
\[
V(z)=\prod_{p |P(z)}\left(1-g(p) \right).
\]
We further suppose for all $w<z$
that
\begin{equation} \label{eq:mult}
\frac{V(w)}{V(z)}=\prod_{\substack{w \le p <z \\p \in \mathcal P}}(1-g(p))^{-1} \le \left( \frac{\log z}{\log w}\right)^{\kappa}\left( 1+O\left( \frac{1}{\log w}\right)\right)
\end{equation}
for some $\kappa >0$.
The constant $\kappa$ is referred to as the \textit{dimension of the
  sieve}.

Our arguments also require  sieve weights. Let
$\Lambda=\{ \lambda_{d}\}_d$, be a sequence of real numbers, where $d$
ranges over square-free integers. The sequence $\Lambda$ is referred
to as an upper bound sieve provided that
\begin{equation} \label{eq:upper}
1_{n=1}=\sum_{d|n} \mu(d) \le \sum_{d|n} \lambda_d, \qquad \forall n \in \mathbb N,
\end{equation}
where $1_{n=1}$ equals one if $n=1$ and equals zero otherwise. We call $\Lambda$  a lower bound sieve if 
\begin{equation} \label{eq:lower}
\sum_{d|n} \lambda_d \le 1_{n=1}, \qquad \forall n \in \mathbb N.
\end{equation}
For a sieve $\Lambda=\{\lambda_d\}$ we use the notation
\begin{equation} \label{eq:sievenotation}
(\lambda \ast 1)(n)=\sum_{d|n} \lambda_d.
\end{equation}
(this will be used to show the existence of primes, or almost
primes with desired properties.)
Additionally, we say that the sieve $\Lambda$ has {\em level} $D$ if
$\lambda_d=0$ for $d>D$.

Given $\kappa>0$ the $\beta$-sieve gives both an upper and lower bound
for $\mathcal S(\mathcal A, \mathcal P, z)$ whenever $s=\log D/\log z$
is sufficiently large in terms of $\kappa$.
The bounds consist of an
error term, which is a sum of the remainder terms $|r_d|$ for $d \le
D$ and a main term $XV(z)F(s)$, $XV(z)f(s)$ (resp.) where $F,f$ are
certain continuous functions with  $0 \le f(s) < 1 < F(s)$.  
For precise definitions, motivation and context we refer the reader to
\cite[Chapter~11]{FI}.  

\begin{theorem}[Cf. {\cite[Theorem~11.13]{FI}}] \label{thm:beta}
Let $D \ge z$ and write $s=\frac{\log D}{\log z}$. Then 
\[
\begin{split}
\mathcal S(A, \mathcal P, z) \le X V(z)\left( F(s)+O(( \log D)^{-1/6}\right)+R(D,z)\\
\mathcal S(A, \mathcal P, z) \ge X V(z)\left( f(s)+O(( \log D)^{-1/6}\right)-R(D,z)
\end{split}
\]
for $s \ge \beta(\kappa)-1$ and $ s\ge \beta(\kappa)$
(resp.), 
where
\[
R(D,z) \le \sum_{\substack{d \le D \\ d|P(z)}} |r_d|.
\]
\end{theorem}
In particular, note that for $\kappa=1/2$, it is well known that
$\beta=1$ (e.g., see \cite[Ch.~14.2]{FI}.)
In our arguments, we will use $\beta$-sieve
weights, which are as defined in \cite{FI} Sections 6.4-6.5. In
particular for these weights we have $|\lambda_d| \le 1$. We will sometimes
refer to the Fundamental Lemma of the Sieve, by which we mean the
following result (see \cite[Lemma 6.11]{FI}.) 

\begin{theorem} \label{thm:fundamental}
Let $\Lambda^{\pm}=\{ \lambda_d^{\pm} \}$ be upper and lower bound (resp.) $\beta$-sieves of level $D$ with $\beta \ge 4 \kappa+1$. Also, let $s=\log D/\log z$. Then for any multiplicative function satisfying \eqref{eq:mult}  and $s \ge \beta+1$ we have
\[
\sum_{d|P(z)} \lambda_d^{\pm} g(d)=V(z)\left(1+O\left(s^{-s/2} \right) \right).
\]
\end{theorem}
 
 We also require the following estimate for the convolution of two
 sieves (see equation (5.97) and Theorem 5.9 of \cite{FI}).

\begin{theorem} \label{thm:conv}
Let $\Lambda_1=\{ \lambda_{d} \}$ and $\Lambda_2=\{ \lambda_{d}^{'} \}$ be upper-bound sieve weights of level $D_1,D_2$ (resp.). 
Also, let $g_1,g_2$ be multiplicative functions satisfying
\eqref{eq:mult} with $\kappa=1$. Then
\[
\bigg|\sum_{\substack{d,e \\ (d,e)=1 }} \lambda_d \lambda_e^{'} g_1(d)g_2(e) \bigg|  \le (4 e^{2\gamma}+o(1)) \prod_{p}(1+h_1(p)h_2(p)) \prod_{j=1}^2  \prod_{p  < D_j} (1-g_j(p))
\]
as $\min\{D_1,D_2\} \rightarrow \infty$,
where for $j=1,2$,
$
h_j(n)=g_j(n)(1-g_j(n))^{-1}
$
and $\gamma$ is Euler's constant.
\end{theorem}
If in addition $g_1(p), g_2(p) \le 1/p$ so that
$h_1(p)h_2(p) \ll 1/p^2$, which will be the case for us, then
\begin{equation} \label{eq:sievebd}
\bigg| \sum_{\substack{d,e \\ (d,e)=1 }} \lambda_d \lambda_e^{'} g_1(d)g_2(e)  \bigg| \le C \prod_{p< D_1}(1-g_1(p))\prod_{p<D_2}(1-g_2(p))
\end{equation}
where $C>0$ is an absolute constant.

\subsubsection{Notation}
\label{sec:notation-sieve} 
We will also use the notation
\[
P_3(z_1,z_2):=\prod_{\substack{z_1 \le p \le z_2 \\ p \equiv 3 \pmod
    4}} p, \qquad \text{ and }\qquad P_3(z):=P_3(3,z). 
\]
Additionally, let 
$1(n)=1_{\mathbb N}(n)=1$
denote the identity function and let $\tau(n)=(1\ast1)(n)=\sum_{d|n} 1$. 
Also, define
\begin{equation} \label{eq:Bdef}
  \mathcal B(x;q,a,\varepsilon)
  :=
  \sum_{\substack{n \le x\\ n \equiv a \pamod q}} (1_{ \mathcal P_{\varepsilon}} \ast 1_{ \mathcal P_{\varepsilon}'})(n)-\frac{1}{\varphi(q)} \sum_{\substack{n \le x \\(n,q)=1}} (1_{ \mathcal P_{\varepsilon}} \ast 1_{ \mathcal P_{\varepsilon}'})(n).
\end{equation}
Further, $\eta, \delta > 0$ will denote small, but fixed real
numbers.
% \change{$\eta, \delta$ comment added}

% $\pi(x)=\#\{p \le x \}$ and for $a,q \in \mathbb N$ with $(a,q)=1$, $\pi(x;q,a)=\#\{ p \le x : p \equiv a \pmod q\}$.

% {\tt what to do with p=2?}\fixme{Can't we just work
% with $ n \equiv 1 \pmod 4$?}

% Finally, let us also state the well-known Bombieri-Vinogradov Theorem
% (cf. \cite[Theorem~9.18]{FI}).\footnote{Note to self
%   (PK): here $q$ does not denote primes; careful about having blanket
%   ``p,q'' always primes convention. \textcolor{blue}{notation needs to be improved throughout the paper,currently trying to fix this!}}
% \begin{theorem} \label{thm:BV}
% For any $A >0$
% \[
% \sum_{q \le Q} \tmop{max}_{(a,q)=1} \bigg|\pi(x;q,a)- \frac{\pi(x)}{\varphi(q)} \bigg| \ll \frac{x}{(\log x)^{A}}
% \]
% where $Q \le \sqrt{x}(\log x)^{-B}$ for some $B(A)>0$ sufficiently large, and the implied constant depends only on $A$.
% \end{theorem}
\subsection{Preliminary lemmas}
\fixmehide{PK: maybe point out to the reader that we almost get asymptotics,
``we obtain upper and lower bounds of the same order of magnitude.}

% Let $Q_0$ be a square-free integer whose prime factors are all $\equiv 1 \pmod 4$.

% \begin{proposition} \label{prop:sieve} Let $\varepsilon>0$ be sufficiently small and fixed. Then for all sufficiently small $\eta>0$  there exists a constant $C=C(Q_0)>0$ such that for $y=x^{\eta}$
% \[
% \sum_{\substack{p \le x \\ p \equiv 1 \pamod 4 \\ |\theta_p| \le \varepsilon \\ (Q_0p+4,P(y))=1}} b(Q_0p+4) \ge C \eta^{1/2} \varepsilon^2 \frac{x}{(\log x) (\log y)}.
% \]
% \end{proposition}

% \begin{remark} Taking $\eta=1/(\log x)$ (note $(2, Q_0p+4)=1$ for $p \neq 2$) this gives
% \begin{equation} \label{eq:sieve1}
% \sum_{\substack{p \le x \\ p \equiv 1 \pamod 4 \\ |\theta_p| \le \varepsilon}} b(Q_0p+4) \ge C \varepsilon^2 \frac{x}{(\log x)^{3/2} }.
% \end{equation}
% \end{remark}
We begin by showing that the difference between the upper and lower
bound sieves is ``small''.
\begin{lemma} \label{lem:switch}  Let $\Lambda^{\pm}=\{
  \lambda_d^{\pm}\}$ be upper and 
  lower bound linear sieves (resp.) each of level $w=x^{\sqrt{\eta}}$
  where $\eta>0$ is sufficiently small, whose sieve weights are
  supported on integers $d$ such that $d|P(y)$,
  where $y=x^{\eta}$ and
  $(d,2Q_0 f_1r_1)=1$; in particular 
  \begin{equation}
    \label{eq:zero-if-not-coprime}
\lambda_{d}^{\pm} = 0 \text{ if $(d,2Q_0 f_1r_1)> 1.$}
  \end{equation}
  Then  
\[
\begin{split}
&\sum_{\substack{n \le x \\  Q_1|Q_0n+4}} \left((\lambda^+ \ast 1)\left(\frac{Q_0n+4}{Q_1}\right)-(\lambda^- \ast 1) \left(\frac{Q_0n+4}{Q_1}\right)\right) (1_{ \mathcal P_{\varepsilon}} \ast 1_{ \mathcal P_{\varepsilon}'})(n) \\
&\qquad \qquad \qquad
\ll
\varepsilon^{2}
\eta^{1/(4\eta^{1/2})-1} \frac{Q_0}{\varphi(Q_0)} \frac{x \log \log x}{\varphi(Q_1) (\log x)^2}+\frac{x}{(\log x)^{10}}.
\end{split}
\]

\end{lemma}
\begin{proof} 
Switching order of summation, it
follows that 
\begin{equation} \label{eq:ubtriv}
\begin{split}
&
\sum_{\substack{n \le x \\  Q_1|Q_0n+4}} \left((\lambda^+ \ast 1)\left(\frac{Q_0n+4}{Q_1}\right)-(\lambda^- \ast 1) \left(\frac{Q_0n+4}{Q_1}\right)\right) (1_{ \mathcal P_{\varepsilon}} \ast 1_{ \mathcal P_{\varepsilon}'})(n)  \\
& \qquad \qquad \qquad = \sum_{\pm} \pm \sum_{\substack{d <w \\ d|P(y) \\ (d,2Q_0f_1r_1)=1}}\lambda_d^{\pm} \sum_{\substack{n \le x \\ Q_0n+4 \equiv 0 \pamod{dQ_1}}} (1_{ \mathcal P_{\varepsilon}} \ast 1_{ \mathcal P_{\varepsilon}'})(n).
\end{split}
\end{equation}
% \change{Commented out ``note to self'' margin note}
%\marginpar{note to self: (2.3) uses $a|b, c|(b/a)$ iff $ac|b$}
The inner sum on the RHS of \eqref{eq:ubtriv} 
equals
\begin{equation} \label{eq:primes}
\begin{split}
&\frac{1}{\varphi(dQ_1)} \sum_{\substack{n \le x \\ (n,dQ_1)=1}} (1_{ \mathcal P_{\varepsilon}} \ast 1_{ \mathcal P_{\varepsilon}'})(n)+\mathcal B\left(x; dQ_1, \gamma, \varepsilon \right)
\end{split}
\end{equation}
where $\gamma$ is the unique reduced residue
$\pmod{dQ_1}$ satisfying
$\gamma \cdot Q_0  \equiv -4 \pmod{dQ_1}$ and $\mathcal B$ is as defined in \eqref{eq:Bdef}.
Also,
\begin{equation} \label{eq:trivialerrorbd1}
\sum_{\substack{n \le x \\ (n,dQ_1)=1}} (1_{ \mathcal P_{\varepsilon}} \ast 1_{ \mathcal P_{\varepsilon}'})(n)=\sum_{\substack{n \le x }} (1_{ \mathcal P_{\varepsilon}} \ast 1_{ \mathcal P_{\varepsilon}'})(n)+O\bigg( \sum_{\substack{p_1p_2 \le x \\ (p_1p_2,dQ_1) \neq 1}}1_{ \mathcal P_{\varepsilon}}(p_1) 1_{ \mathcal P_{\varepsilon}'}(p_2) \bigg). 
\end{equation}
Since $dQ_1 \le x^{1/9}$ (as $\eta$ is small) and $p_2 \le x^{1/9}$ the contribution to the error term from $p_1p_2 \le x$ with $p_1|(p_1p_2,dQ_1)$ is $\ll \sum_{p_2 \le x^{1/9}} \sum_{p_1 \le x^{1/9} } 1 \ll x^{2/9}$. Also, since $p_2 \ge (\log x)^{B_0}$
\begin{equation} \label{eq:trivialerrorbd2}
\sum_{\substack{p_1p_2 \le x \\ (p_1p_2,dQ_1) = p_2}}1_{ \mathcal P_{\varepsilon}}(p_1) 1_{ \mathcal P_{\varepsilon}'}(p_2) \le \sum_{\substack{p_2|dQ_1 \\ p_2 \ge (\log x)^{B_0}}} \sum_{p_1 \le x/p_2} 1 \ll \frac{x}{\log x} \sum_{\substack{p_2|dQ_1 \\ p_2 \ge (\log x)^{B_0}}} \frac{1}{p_2} \ll \frac{x (\log \log x)}{(\log x)^{B_0}}.
\end{equation}

Hence, using \eqref{eq:primes}, \eqref{eq:trivialerrorbd1}, \eqref{eq:trivialerrorbd2} along with the Fundamental
Lemma of the Sieve (see Theorem \ref{thm:fundamental} and recall $|\lambda_d|\le 1$) with $g(d)=\varphi(Q_1)/\varphi(Q_1 d)$\footnote{Note that $g$ is multiplicative on the set of square-free $d$ with $(d,f_1r_1)=1$.},
and
$s=\log w/\log y=\eta^{-1/2}$ we have that
\begin{equation} \label{eq:summedup}
\begin{split}
&\sum_{\substack{d <w \\ d|P(y) \\ (d,2Q_0)=1}}\lambda_d^{\pm} \sum_{\substack{n \le x \\ Q_0n+4 \equiv 0 \pamod{dQ_1}}} (1_{ \mathcal P_{\varepsilon}} \ast 1_{ \mathcal P_{\varepsilon}'})(n)\\
&=\frac{1}{\varphi(Q_1)} \sum_{\substack{n \le x }} (1_{ \mathcal P_{\varepsilon}} \ast 1_{ \mathcal P_{\varepsilon}'})(n)   \prod_{\substack{p \le y \\ (p , 2Q_0f_1r_1 )=1}}\left(1-\frac{\varphi(Q_1)}{\varphi(Q_1p)} \right)(1+O(\eta^{1/(4\eta^{1/2})}))\\
& \qquad \qquad \qquad +O\bigg( \sum_{\substack{d < w \\ (d,2)=1}} \left|\mathcal B \left(x; dQ_1, \gamma, \varepsilon \right)\right|\bigg)+O\left( \frac{x \log \log x}{(\log x)^{B_0-1}}\right).
\end{split}
\end{equation}
Applying Theorem \ref{thm:bv} from the appendix, since $w=x^{\sqrt{\eta}} < x^{1/2-o(1)}$ we get that
\[
\sum_{\substack{d < w \\ (d,2)=1}} \left|\mathcal B \left(x; dQ_1, \gamma, \varepsilon \right)\right| \ll \frac{x}{ (\log x)^{10}}.
\]
Using the two estimates above in \eqref{eq:ubtriv} (note the main terms in \eqref{eq:summedup} are the same for each of the sieves $\Lambda^{\pm}$ so they cancel in \eqref{eq:ubtriv}) and applying \eqref{eq:asymp} (with $q=1$) from the appendix to estimate the sum over $n$, completes the proof upon noting that
\[
  \prod_{\substack{p \le y \\ (p , 2Q_0f_1r_1)=1}}\left(1-\frac{\varphi(Q_1)}{\varphi(Q_1p)} \right) \asymp \frac{Q_0}{\varphi(Q_0) \log y} = \frac{Q_0}{\varphi(Q_0) \eta \log x}.
\]
\end{proof}

We next give a lower bound on the upper bound sieve, which together
with Lemma~\ref{lem:switch} is strong enough (given suitable parameter
choices) to show the existence of infinitely many
integers with \textit{exactly} two prime factors with the desired
properties.
\begin{lemma} \label{lem:lbsieve} Let  $w = x^{\sqrt{\eta}}$,
  $y=x^{\eta}$, and  $\Lambda^{+}$ be as in Lemma
  \ref{lem:switch}. Let $\delta> 3 \sqrt{\eta}>0$ and
  $z=x^{\frac12-\delta}$. Then there exists a constant $C_1>0$ such
  that
\[
\sum_{\substack{n \le x \\ (Q_0n+4, P_3(y,z))=1 \\ Q_1|Q_0n+4}} (1_{ \mathcal P_{\varepsilon}} \ast 1_{ \mathcal P_{\varepsilon}'}) (n) (\lambda^+ \ast 1)\left(\frac{Q_0n+4}{Q_1}\right) \ge C_1 \frac{\varepsilon^2 \delta^{1/2}}{\eta^{1/2}}  \frac{Q_0}{\varphi(Q_0)} \frac{x \log \log x}{\varphi(Q_1)(\log x)^2}.
\]
\end{lemma}
\begin{proof}
Consider the sifting sequence 
\[
\mathcal A=\bigg\{ (1_{ \mathcal P_{\varepsilon}} \ast 1_{ \mathcal P_{\varepsilon}'})\bigg(\frac{m-4}{Q_0}\bigg) (\lambda^+ \ast 1)\left(\frac{m}{Q_1}\right)  : Q_1|m, Q_0|m-4 \bigg\}
\]
and primes $\mathcal P=\{ p \ge y : p \equiv 3 \pmod 4 \}$.
Recalling (\ref{eq:zero-if-not-coprime}), we may write
\begin{equation} \label{eq:X}
\begin{split}
X
=& \sum_{\substack{e<w \\ e|P(y)}} \frac{\lambda_e^+}{\varphi(eQ_1)} \sum_{\substack{n \le x  \\ (n, Q_1e)=1}} (1_{ \mathcal P_{\varepsilon}} \ast 1_{ \mathcal P_{\varepsilon}'})(n) \\
=&\sum_{\substack{n \le x  \\ (n, Q_1)=1}} (1_{ \mathcal P_{\varepsilon}} \ast 1_{ \mathcal P_{\varepsilon}'})(n) \sum_{\substack{e<w \\ e|P(y) \\ (e,2Q_0f_1r_1n)=1}} \frac{\lambda_e^+}{\varphi(eQ_1)} \gg \varepsilon^2 \frac{Q_0}{\varphi(Q_0)}  \frac{x \log \log x}{ \varphi(Q_1) (\log y) (\log x)},
\end{split}
\end{equation}
where the lower bound follows from the Fundamental Lemma of the Sieve
(see  \eqref{eq:summedup} and take $D=w$, $z=y$ in Theorem 
\ref{thm:fundamental} and note that we then have  $s=\eta^{-1/2}$) along with
prime number theorem for Gaussian primes in sectors to evaluate the
sum over $n$ (see  \eqref{eq:PNT},
\eqref{eq:asymp} in the Appendix).

For $d|P_3(y,z)$ note that $(d, eQ_0 Q_1)=1$ for $e$ such that $p|e
\Rightarrow p< y$, and $(1_{ \mathcal P_{\varepsilon}} \ast 1_{
  \mathcal P_{\varepsilon}'})(n)=0$ if $(d,n) \neq 1$. It follows
that (cf. (\ref{eq:Ad1}) and (\ref{eq:Ad2}) for the definition of $A_{d}$)
\[
\begin{split}
A_d(Q_0x+4)&=\sum_{\substack{n \le x \\ Q_1 | Q_0n+4 \\ Q_0n+4 \equiv 0 \pamod d}} (1_{ \mathcal P_{\varepsilon}} \ast 1_{ \mathcal P_{\varepsilon}'})(n) (\lambda^+ \ast 1)\left(\frac{Q_0n+4}{Q_1}\right) \\
&= \sum_{\substack{e< w \\ e|P(y)}} \lambda_e^{+} \sum_{\substack{ n \le x \\ Q_0n+4 \equiv 0 \pamod{eQ_1} \\ Q_0n+4 \equiv 0 \pamod d}} (1_{ \mathcal P_{\varepsilon}} \ast 1_{ \mathcal P_{\varepsilon}'})(n) \\
&=\sum_{\substack{e< w \\ e|P(y)}} \frac{\lambda_e^{+}}{\varphi(deQ_1)} \sum_{\substack{n \le x \\ (n, Q_1 e)=1}} (1_{ \mathcal P_{\varepsilon}} \ast 1_{ \mathcal P_{\varepsilon}'})(n)+r_d =\frac{1}{\varphi(d)}X+r_d,
\end{split}
\]
where 
\[
r_d \ll \sum_{\substack{e<w \\ (e,2)=1}}
\left|\mathcal B(x;deQ_1, \gamma, \varepsilon)\right|
\]
and $\gamma$ is the unique residue class $\pmod{deQ_1}$ with
$Q_0 \gamma \equiv -4 \pmod{eQ_1}$ and $Q_0\gamma \equiv -4 \pmod d$;
also note that $(d,eQ_1)=1$ and $\mathcal B$ is as in \eqref{eq:Bdef}.

Hence, the half-dimensional Rosser-Iwaniec sieve, Theorem \ref{thm:beta}, gives for any $D \ge z$ with $s=\log D/\log z$
\begin{equation} \label{eq:lbsieve}
\begin{split}
\sum_{\substack{n \ge 1 \\ (Q_0n+4, P_3(y,z))=1 \\ Q_1|Q_0n+4}} (1_{
  \mathcal P_{\varepsilon}} \ast 1_{ \mathcal P_{\varepsilon}'}) (n)
(\lambda^+ \ast 1)\left(\frac{Q_0n+4}{Q_1}\right)& \\ \ge X
V(z)\left(f(s)+O\left(\frac{1}{(\log D)^{1/6}}\right)\right)
&-\sum_{\substack{d < D \\ d|P_3(y,z)}} |r_d|
\end{split}
\end{equation}
where 
\begin{equation} \label{eq:V}
V(z)=\prod_{\substack{y \le p \le z \\ p \equiv 3 \pmod 4}} \left( 1-\frac{1}{p-1} \right)\gg  \sqrt{\frac{\log y}{\log z}} \gg \eta^{1/2}.
\end{equation}
Taking $D=z^{1+\delta}$, so $s=1+\delta$, we have by Theorem
\ref{thm:bv}, which is proved in the appendix, that
(taking
$q=edQ_{1}$)
\begin{equation} \label{eq:R}
\sum_{\substack{d < D \\  d|P_{3}(y,z)}} |r_d| \ll \sum_{ \substack{q
    < DQ_1w \\ (q,2)=1}}\left( \tau(q) \max_{(a,q)=1} |\mathcal
  B(x;q,a, \varepsilon)| \right)\ll \frac{x}{(\log x)^3}. 
\end{equation}
Here note that $DQ_1w < x^{\frac12-\frac{\delta}{2}+\sqrt{\eta}}<
x^{\frac12-\frac{\delta}{6}}$ and the contribution of the divisor
function is handled by using Cauchy-Schwarz along with the trivial
bound $|\mathcal B(x;q,a, \varepsilon)| \ll x/q $.
Also note that $f(t) \sim 2 \sqrt{\frac{e^{\gamma}}{\pi}} \cdot
\sqrt{t-1}$ as $t \rightarrow 1^+$ (see the equation after (14.3) of
\cite{FI}), so $f(s)=f(1+\delta) \gg \sqrt{\delta}$. 
Using this along with \eqref{eq:X}, \eqref{eq:V}, and \eqref{eq:R} in \eqref{eq:lbsieve} completes the proof.

\end{proof}

% \begin{lemma} \label{lem:babd}
% We have for $N>y^{10}$
% \[
% \sum_{\substack{a \le N \\ (a, P(y))=1 }} b(a)  \ll \frac{N}{\sqrt{\log N \log y}}.
% \]
% \end{lemma}

% \begin{proof}
% Using upper bound sieves $\Lambda=\{\lambda_d\}$ and $\Lambda'=\{\lambda_e'\}$ of level $N^{1/3}$ it follows that
% \[
% \begin{split}
% \sum_{\substack{a \le N \\ (a, P(y))=1 }} b(a) \le & \sum_{\substack{a \le N \\ (a, P_1(y))=1 \\ (a,P_3(N^{1/10}))=1}} 1 \\
% \le & \sum_{d| P_1(y), e|P_3(N^{1/10})} \lambda_d \lambda_e' \sum_{\substack{a \le N  \\ d| a \\ e |a }} 1 \\
% \end{split}
% \]
% Since $d|P_1(y)$ and $e|P_3(N^{1/10})$ it follows that $(e,d)=1$. So the sum on the RHS above equals
% \[
% \begin{split}
% =& N \sum_{d| P_1(y), e|P_3(N^{1/0})} \frac{\lambda_d \lambda_e}{de}+O(N^{2/3})\\
% \ll & N \prod_{\substack{ p \le y \\ p \equiv 1 \pamod 4}}\left( 1-\frac1p\right)\prod_{\substack{ p \le N^{1/10} \\ p \equiv 3 \pamod 4}}\left( 1-\frac1p\right) \ll \frac{N}{\sqrt{\log N \log y}}.
% \end{split}
% \]
% \end{proof}

\subsection{The Proof of Proposition \ref{prop:sieve}}
We first require a Brun-Titchmarsh type bound for primes in narrow sectors.
\begin{lemma}  \label{lem:fixed}
Let $Q, q\le x^{2/3-o(1)}$ be odd. Then
\[
\sum_{ \substack{p=a^2+b^2 \le x \\ |\arctan(b/a)| \le \varepsilon \\
    qp+4=Q p_1, \, p_1 \text{ prime}  }} 1 \ll
\varepsilon \frac{q}{\varphi(q)}\frac{x}{\varphi(Q)(\log x)^2}.  
\]
\end{lemma}
\begin{remark}
The point of the lemma is that it holds for large moduli
$Q>x^{1/2}$. To accomplish this we use asymptotic estimates for
Gaussian integers $\alpha=a+ib$ with $N(\alpha) \le x$ and $N(\alpha)
\equiv a \pmod Q$ \textbf{and} $|\arg(\alpha)|\le \varepsilon$, where
$N(\alpha)=\alpha \overline{\alpha}$ is the norm of $\alpha$.
Details are given in Appendix, cf. section \ref{sec:gauss-integ-sect}.
\end{remark}

The main step in the proof of Proposition \ref{prop:sieve} is the following lemma.

\begin{lemma} \label{lem:reverse} 
Let $z=x^{\frac12-\delta}$ where $\delta>0$ 
is sufficiently small and $ y = x^{\eta}$ with $0< \eta < 1/3$. 
There exists a constant $C_2>0$ such that
\[
\sum_{\substack{n \le x \\ Q_1|Q_0n+4 \\ (\frac{Q_0n+4}{Q_1},P(y)P_3(y,z))=1}} (1_{ \mathcal P_{\varepsilon}} \ast 1_{ \mathcal P_{\varepsilon}'})(n) = \sum_{\substack{n \le x \\ Q_1|Q_0n+4 \\(\frac{Q_0n+4}{Q_1},P(y))=1 \\ p|Q_0n+4 \Rightarrow p \equiv 1 \pamod 4}} (1_{ \mathcal P_{\varepsilon}} \ast 1_{ \mathcal P_{\varepsilon}'})(n) +R
\]
where
\[
0 \le R \le C_2 \cdot \varepsilon^2 \cdot \frac{\delta^{3/2}}{\eta^{1/2}} \frac{Q_0}{\varphi(Q_0)} \cdot \frac{x \log \log x}{\varphi(Q_1)(\log x)^2}.
\]
\end{lemma}
\begin{proof}
By construction for $\ast 1_{ \mathcal P_{\varepsilon}'})(n) \neq 0$, $Q_0n+4 \equiv 1 \pmod 4$ and $Q_1 \equiv 1 \pmod 4$ so that $(Q_0n+4)/Q_1 \equiv 1 \pmod 4$ and must have an even number of prime factors which are congruent to $3 \pmod 4$. Since $z>x^{1/4}$ the integers which contribute to $R$ must have precisely two such prime factors. 
Dropping several conditions on the integers $n$ which contribute to
$R$, it follows  that $R$ is bounded by the number of integers
$n=p_1p_2 \le x$, $( 1_{ \mathcal P_{\varepsilon}} \ast 1_{ \mathcal
  P_{\varepsilon}'})(n) \neq 0$ such that $(Q_0n+4)/Q_1=aq_1q_2$ where
$b(a)=1$, $(a,P(y))=1$, $q_1 \equiv q_2 \equiv 3 \pmod 4$ and
$q_1,q_2$ are primes with $z< q_1, q_2 \le 2Q_0x/Q_1$ so $a \le
2Q_0x/(Q_1z^2)$. By symmetry, it suffices to consider the terms with
$q_1 \le q_2$. We get that
\begin{equation} \label{eq:Rfirst}
R \le 2  \sum_{p_2 \le x^{1/9} } 1_{\mathcal P_{\varepsilon}'}(p_2)
\sum_{\substack{ a \le \frac{2Q_0 x}{Q_{1}z^2} \\ (a,P(y))=1}} b(a)
\sum_{z< q_1 \le \sqrt{\frac{2Q_0 x}{aQ_1}}}
\sum_{q_1  \le q_2 \le
  2Q_0x/Q_1} \sum_{\substack{p_1 \le x/p_2 \\ Q_0p_1p_2+4=aq_1q_2Q_1}}
1_{\mathcal P_{\varepsilon}}(p_1). 
\end{equation}
Applying Lemma \ref{lem:fixed} with $q=Q_0p_2$
and $Q=aq_1Q_1$ 
\begin{equation} \label{eq:primebd}
 \sum_{\substack{p_1 \le x/p_2 \\ Q_0p_1p_2+4=aq_1q_2Q_1}} 1_{\mathcal P_{\varepsilon}}(p_1) \ll \varepsilon \frac{Q_0}{\varphi(Q_0)}\frac{x}{ \varphi(aQ_1) q_1p_2 (\log x)^2}.
\end{equation}
Note that $x/p_2 \ge x^{8/9}$ and
$Q_0p_2, aq_1Q_1 \le \left( \frac{x}{p_2}\right)^{2/3-o(1)}$, for $\delta>0$ sufficiently
small so the application of Lemma \ref{lem:fixed} is valid.

We claim that
% \marginpar{This partial summation needs to be justified more carefully. Could probably just use Lemma \ref{lem:NT}.}
\begin{equation} \label{eq:Rbd4}
\sum_{\substack{a \le \frac{2Q_0x}{Q_1z^2} \\ (a,P(y))=1}} \frac{b(a)}{\varphi(a)}  \ll  \sqrt{\frac{\log x/z^2}{\log y}},
\end{equation}
which we will justify below.
Additionally,
\begin{equation} \label{eq:Rbd5}
\sum_{z<q_1 \le \sqrt{\frac{2Q_0 x}{aQ_1}}} \frac{1}{q_1}
\sim \log \frac{\log \sqrt{\frac{2Q_0 x}{aQ_1}}}{\log z}
\ll \frac{\log \frac{ x}{z^2}}{\log z}+\frac{\log Q_0}{\log z}
\ll \frac{\log \frac{ x}{z^2}}{\log z}
\ll \delta.
\end{equation}
Therefore, using \eqref{eq:primebd}, \eqref{eq:Rbd4}, and \eqref{eq:Rbd5} in  \eqref{eq:Rfirst} we conclude that
\[
\begin{split}
 R \ll & \varepsilon  \cdot \frac{Q_0}{\varphi(Q_0)} \cdot \frac{x \log x/z^2}{\varphi(Q_1) (\log x)^2 \log z}  \sqrt{\frac{\log x/z^2}{\log y}}  \sum_{p_2 \le x^{1/9} } \frac{  1_{\mathcal P_{\varepsilon}'}(p_2)}{p_2} \\
 \ll & \varepsilon^2 \cdot  \frac{\delta^{3/2}}{\eta^{1/2}}\cdot \frac{ Q_0}{  \varphi(Q_0)} \cdot \frac{x \cdot  \log \log x }{\varphi(Q_1) (\log x)^2}  
 \end{split}
\]
as desired.

It remains to justify \eqref{eq:Rbd4}. Let $F(n)$ be the completely
multiplicative function defined by $F(p)=1$ if $p \ge y$ and zero
otherwise. Then for all $t \ge y$, it follows from basic estimates for
multiplicative functions (see (1.85) of \cite{IK}) that 
\[
\begin{split}
\sum_{\substack{n \le t \\ (n,P(y))=1}} b(n) \frac{n}{\varphi(n)} \le& \sum_{\substack{n \le t }} b(n) \frac{n}{\varphi(n)} F(n) \\
\ll & \frac{t}{\log t} \prod_{ p \le t} \left(1+ \frac{b(p)F(p)}{p-1} \right) 
\ll \frac{t}{\sqrt{\log t \log y}}.
\end{split}
\]
For $1 \le t \le y$ the sum on the LHS is empty so the bound is true in that case as well. Hence, $\eqref{eq:Rbd4}$ follows from this estimate along with partial summation.
\end{proof}

\begin{proof}[Proof of Proposition \ref{prop:sieve}] 
  Let $\delta$ be sufficiently small in terms of $\eta$, $C_1$ and $C_2$.  Applying
  the inequality \eqref{eq:lower} for a lower bound sieve (also recall
  our notation \eqref{eq:sievenotation}) along with
  Lemmas \ref{lem:switch} and \ref{lem:lbsieve}, using a lower bound sieve to take care of the condition
 $( \frac{Q_0n+4}{Q_1} , P(y)) = 1$, we have
  that
\begin{equation} \label{eq:sievelast}
\begin{split}
\sum_{\substack{n \le x \\ Q_1 |Q_0n+4 \\ (\frac{Q_0n+4}{Q_1},P(y)P_3(y,z))=1 }} (1_{ \mathcal P_{\varepsilon}} \ast 1_{ \mathcal P_{\varepsilon}'})(n)
\ge& \sum_{\substack{ n \le x \\ Q_1 |Q_0n+4\\ (Q_0n+4,P_3(y,z))=1}} (1_{ \mathcal P_{\varepsilon}} \ast 1_{ \mathcal P_{\varepsilon}'})(n) (\lambda^- \ast 1)\left(\frac{Q_0n+4}{Q_1}\right) \\
=&   \sum_{\substack{ n \le x \\ Q_1 |Q_0n+4\\ (Q_0n+4,P_3(y,z))=1}} (1_{ \mathcal P_{\varepsilon}} \ast 1_{ \mathcal P_{\varepsilon}'})(n) (\lambda^+ \ast 1)\left(\frac{Q_0n+4}{Q_1}\right)\\
&\qquad \qquad +O\left(\varepsilon^2 \eta^{1/(4\eta^{1/2})-1} \frac{Q_0}{\varphi(Q_0)} \frac{x \log \log x}{\varphi(Q_1)(\log x)^2} \right) \\
\ge& C_1 \frac{\varepsilon^2 \delta^{1/2}}{\eta^{1/2}} \frac{Q_0}{\varphi(Q_0)} \frac{x \log \log x}{\varphi(Q_1)(\log x)^2}\left(1+O\left( \frac{ \eta^{\frac{1}{4\eta^{1/2}}-\frac12}}{\delta^{1/2}} \right)\right) .
\end{split}
\end{equation}
 Choosing $\eta$ sufficiently small in terms of $\delta$ (which is fixed) the $O$-term above is $ \le 1/2$ in absolute value.
Therefore, by \eqref{eq:sievelast} along with Lemma \ref{lem:reverse} it follows that
\begin{equation} \notag
\begin{split}
\sum_{\substack{n \le x \\ Q_1 |Q_0n+4 \\ (\frac{Q_0n+4}{Q_1},P(y)P_3(y,z))=1 \\ p|Q_0n+4 \Rightarrow p \equiv 1 \pamod 4 }} (1_{ \mathcal P_{\varepsilon}} \ast 1_{ \mathcal P_{\varepsilon}'})(n) 
% & \qquad \ge \frac{C_1}{2}\frac{\varepsilon \delta^{1/2}}{\eta^{1/2}} \frac{Q_0}{\varphi(Q_0)} \frac{x \log \log x}{\varphi(Q_1)(\log x)^2}- C_2 \cdot \frac{Q_0}{\varphi(Q_0)} \cdot \frac{x \log \log x}{\varphi(Q_1)(\log x)^3}\left(  \log \frac{x}{z^2}\right) \cdot  \sqrt{\frac{\log x/z^2}{\log y}} \\
 \ge \left(\frac{C_1}{2} \frac{\varepsilon^2 \delta^{1/2}}{\eta^{1/2}}-\frac{C_2 \varepsilon^2 \delta^{3/2}}{\eta^{1/2}} \right) \frac{Q_0}{\varphi(Q_0)} \frac{x \log \log x}{\varphi(Q_1)(\log x)^2}.
\end{split}
\end{equation}
The term $\left(\frac{C_1}{2} \delta^{1/2}-C_2  \delta^{3/2} \right)$
is positive for  $\delta$ sufficiently small in terms of $C_1$ and
$C_2$. Also $b(Q_0n+4) =1$
for $n$ such that all the prime factors of $Q_0n+4$ are congruent to
$1 \pmod 4$. This 
completes the proof.
\end{proof}
\section{Truncating the spectral equation} \label{sec:truncate}

% \footnote{\textcolor{blue}{Fixme: SL perhaps the notation should be
%     changed here $B$, $X$, $y$ and $\delta$ are used in the previous
%     section} SL: I did this earlier, the notation should be somewhat
%   better now. PK: looks good.}
% Let
% \[
% S_1=\{ m_0 p \le x : p=a^2+b^2 \text{ and } |\arctan(b/a)| \le \varepsilon \} 
% \]
% and
% \[
% S_2=\{ m_0 p \le x : p \equiv 3 \pmod 4 \text{ and } p+2 \in \mathcal N_2 \cap \mathcal P_2 \}.
% \]
% For $j=1,2$ write
% \[
% \mathcal E_j(x)=\{ \lambda \le x : m_{\lambda} \in S_j \}.
% \]

In this section we show that it is possible to achieve a very short truncation of the spectral equation which holds for almost all new eigenvalues.

\begin{theorem} \label{thm:spectral} 
Let $A \ge 1$. Then for
  $B=B(A)$ sufficiently large  we have for every eigenvalue $\lambda_n
  \in \Lambda_{new} \cap [1,x]$ except those outside an exceptional
  set of size $O(x/(\log x)^A)$ that 
%   \fixme{It should be $\sum_{|m-n|
%       \ge \frac{m}{X}  (\log X)^B}$ in the sum, no? \textbf{Yes, fixed, I believe it was correct as before given the definition of $\lambda$, but this relies on the strong coupling quantization used in \eqref{eq:spectral} }}\fixme{Question:
%     any hope of taking $B$ close to $1/2$ and still getting some sort
%     of saving? \textbf{This would be nice and is perhaps do-able. It would use some of the new work of Matomaki-Radziwill and would probably require some significant extra work. I'll ask Radziwill about this}}
\begin{equation} \label{eq:initial2}
\sum_{m:|m-n| \le \frac{n}{x} (\log x)^B} \frac{r(m)}{m-\lambda_n}
=
\begin{cases}
\pi \log \lambda_n+O(1)&  \text{ in the weak coupling quantization},\\
\frac{1}{\alpha}+O(1) & \text{ in the strong coupling quantization}.
\end{cases}
\end{equation}

\end{theorem}

The above theorem is proved by capturing cancellation in the spectral equation even at very small scales, for almost all new eigenvalues. This is done by showing that the average behavior of sums of $r(n)$ over even very short intervals is fairly regular.
% This is done, essentially, by using Plancherel's theorem, which allows us to show that the average behavior of $r(n)$ over almost all short intervals is fairly regular. % \begin{proposition}
% For any $B \ge 1$ there exists $A=A(B)$ sufficiently large so that
% for all but $O(\mathcal E_j(x)/(\log x)^B)$ eigenvalues $\lambda \in \mathcal E_j(x)$ we have
% \[
% \sum_{1 \le |k| \le x^{1/3}} \frac{r(m_{\lambda}+k)}{k}=\frac{r(m_{\lambda}+2j)}{2j}+\sum_{(\log x)^{\frac12-\varepsilon} \le |k| \le (\log x)^A} \frac{r(m_{\lambda}+k)}{k}+O\left(\frac{1}{(\log x)^{\varepsilon}} \right).
% \]
% \end{proposition}

% \begin{lemma}
% For all but $O\left(\mathcal E_j(x)/(\log x)^{\varepsilon} \right)$ eigenvalues $\lambda \in \mathcal E_j(x)$ we have
% $
% r(m_{\lambda}+h) = 0
% $
% for all $1 \le |h| \le (\log x)^{\frac12-\varepsilon}$, except $h=2j$.
% \end{lemma}
% \begin{proof} We apply Chebyshev's inequality to see that
% \[
% \begin{split}
% \# \{ \lambda \in \mathcal E_j(x) : \sum_{1 \le |h| \le (\log x)^{\frac12-\varepsilon}} b(m_{\lambda}+h) \ge 1 \}
% \le & \sum_{1 \le |h| \le (\log x)^{\frac12-\varepsilon}} \sum_{ \lambda \in \mathcal E_j(x)} b(m_{\lambda}+h) \\
%  & \ll  \sum_{1\le |h| \le (\log x)^{\frac12-\varepsilon}} \frac{ h}{\varphi(h)} \frac{\# \mathcal E_j(x)}{\sqrt{\log x}}
% \end{split}
% \]

% \end{proof}

\begin{lemma} \label{lem:selberg}
Let $x\ge3$  and $3 \le L \le x $. Also, let $h(x)=x/L$.  Then
\begin{equation} \label{eq:varbd}
\frac{1}{x}
\sum_{ \ell \le x} \bigg| \sum_{\ell \le n \le \ell+h(\ell)} r(n)-\pi h(\ell) \bigg|^2 \ll h(x) (\log x)^2.
\end{equation}
\end{lemma}

\begin{proof}
We repeat a classical argument, which was used by Selberg \cite{Selberg} to study primes in short intervals. Consider
\[
\zeta_{\mathbb Q(i)}:= \frac14 \sum_{n \ge 1} \frac{r(n)}{n^s}=L(s,\chi_4) \zeta(s)  \qquad \tmop{Re}(s)>1,
\]
where $L(s,\chi_4)$ is the Dirichlet $L$-function attached to the non-trivial Dirichlet character $\pmod 4$, and $\zeta(s)$ denotes the Riemann zeta-function. Note $L(1, \chi_4)=\pi/4$.
Applying Perron's formula, then shifting contours to
$\tmop{Re}(s)=1/2$ (which is valid since the it is well-known that\fixmehide{Check, SL this is the convexity bound, we could add a reference} $\zeta_{\mathbb
  Q(i)}(\sigma+it) \ll t^{1-\sigma+o(1)}$, for $0 \le \sigma \le 1$)
and picking up a simple pole at $s=1$ we see that for $v,v+v/L \notin
\mathbb Z$
\[
\begin{split}
\sum_{v \le n \le v+\frac{v}{L}} r(n)=& \frac{1}{2\pi i} \int_{(2)}
4\zeta_{\mathbb Q(i)} (s) \frac{(v+\frac{v}{L})^s-v^s}{s} \, ds \\
=& 4 L(1,\chi_4) \cdot \frac{v}{L}+\frac{v^{1/2}}{2\pi} \int_{\mathbb R} 4 \zeta_{\mathbb Q(i)}(\tfrac12+it) \frac{(1+\frac{1}{L})^{\frac{1}{2}+it}-1}{\frac12+it} \cdot e^{i t \log v} \, dt.
\end{split}
\]
Notice that the integral on the RHS is a Fourier transform. Writing
$\nu=\log(1+\frac1L)$, making a change of variables $x=e^{\tau}$
and then applying Plancherel's Theorem yields
\[
\begin{split}
\frac{1}{x^2} \int_1^{x} \bigg(\sum_{v \le n \le v+\frac{v}{L}} r(n)- \pi \cdot \frac{v}{L} \bigg)^2 \, dv \le &  \int_{\mathbb R} \bigg(\sum_{e^{\tau} \le n \le e^{\tau+\nu}} r(n)- \pi \cdot \frac{e^{\tau}}{L} \bigg)^2  \frac{d\tau}{e^{\tau}} \\
=&\frac{8}{\pi} \int_{\mathbb R} |\zeta_{\mathbb Q(i)}(\tfrac12+it)|^2 |w_{\nu}(\tfrac12+it)|^2 \, dt
\end{split}
\]
where $w_{\nu}(s)=(e^{\nu s}-1)/s \ll \min\{ \nu, 1/(1+|t|) \}$ uniformly for $\frac14 \le \tmop{Re}(s) \le 1$.
To estimate the integral on the RHS we apply
the well-known bound 
\[
\int_0^{T} |\zeta_{\mathbb Q(i)} (\tfrac12+it)|^2 \, dt \ll T (\log T)^2
\]
(see the introduction of \cite{Mu}).
Hence we see that
\[
\begin{split}
\int_{\mathbb R} |\zeta_{\mathbb Q(i)}(\tfrac12+it)|^2 |w_{\nu}(\tfrac12+it)|^2 \, dt \ll &  \nu^2 \int_{|t| \le 1/\nu} |\zeta_{\mathbb Q(i)} (\tfrac12+it)|^2 \, dt+\int_{|t| \ge 1/\nu} |\zeta_{\mathbb Q(i)} (\tfrac12+it)|^2 \frac{dt}{t^2} \\
\ll & \nu (\log 1/\nu)^2 \ll \frac{1}{L} (\log L)^2.
\end{split}
\]
Combining the estimates above we conclude that for $h=h(x)=x/L$
\begin{equation} \label{eq:plancherel}
\frac{1}{x} \int_{x}^{2x} \bigg( \sum_{v \le n \le v+h(v)} r(n)-\pi h(v) \bigg)^2 \, dv \ll h(x) (\log x)^2.
\end{equation}

% By Chebyshev's inequality
% \# \left \{ X \le \ell \le 2X :
%  F(\ell) \ge \sqrt{h(X)} (\log X)^A 
%  \right \}

We will now bound the sum over integers $  \ell \le x$ on the LHS of
\eqref{eq:varbd} in terms of an integral over $1 \le v \le x$. Let 
\[
F(v)=\sum_{v \le n \le v+h(v)} r(n)- \pi h(v)
\]
and let $v_{\ell} \in [\ell, \ell+1]$ be a point where the minimum 
of
$|F(v)|$ 
on $[\ell, \ell+1]$ is achieved. Observe that
\[
F(\ell)=F(v_{\ell})+O\left( r(\ell)+r(\ell^*)+1 \right)
\]
where $\ell^*=\lfloor \ell+1+h(\ell+1) \rfloor $. Hence,
\[
\begin{split}
\frac{1}{x}\sum_{ \ell \le x} F(\ell)^2 &\ll
\frac{1}{x}\sum_{ \ell \le x} F(v_{\ell} )^2+\frac{1}{x}\sum_{\ell \le x }(r^2(\ell)+r^2(\ell^*))+1 \\ 
& \ll  
\frac{1}{x}\int_{1}^{x} F(x)^2 \, dx+ \log x \ll  h(x) (\log x)^2,
\end{split}
\]
where the last bound follows from \eqref{eq:plancherel}.

\end{proof}

\begin{lemma} \label{lem:partial}
Let $A \ge 3$ and $x,Y \ge 3$. Then
for all but $\ll x/(\log x)^{A}$ integers $m \in [ 1, x]$ we
have
\[
\bigg| \sum_{Y \frac{m}{x} < k \le x^{1/2} \frac{m}{x}} \frac{r(m+k)-r(m-k)}{k}
\bigg| \le \frac{(\log x)^{3A}}{\sqrt{Y}}.
\]
\end{lemma}

\begin{proof}

Let
\[
R_m(t)=\sum_{1 \le k \le t} (r(m+k)-r(m-k)).
\]
% Applying the previous lemma we have for all but
% $\ll X/(\log X)^{2A-2}$ integers $X \le m \le 2X$ that
% \[
% R_m(u)=\sum_{m \le n \le m+u} r(n) - \sum_{m-u \le n \le m} r(n) \ll \sqrt{u} (\log X)^A,
% \]
% for $u=m/L$, uniformly for all  $1 \le L\le X/(\log X)^A$ for each $m$ outside our exceptional set (i.e. the set of integers satisfy this equation does not depend on $L$), so that this holds with for all $t$ with $(\log  X)^A \le t \le 2X$. 
It suffices to consider $m \in [x/(\log x)^A,x]$.
Hence, by summation by parts for each integer $m \in [x/(\log x)^A,x]$ we have that
\[
\begin{split}
\sum_{Y \frac{m}{x} < k \le x^{1/2} \frac{m}{x}} \frac{r(m+k)-r(m-k)}{k}=&\frac{R_m(x^{1/2} \frac{m}{x})}{x^{1/2} \frac{m}{x}}- \frac{R_m(Y \frac{m}{x})}{Y \frac{m}{x}}+\int_{Y \frac{m}{x}}^{x^{1/2} \frac{m}{x}} \frac{R_m(t)}{t^2} \, dt.
\end{split}
\]

Using this along with Chebyshev's inequality and the elementary
inequality $(|a|+|b|+|c|)^{2} \le 3^2
(a^2+b^2+c^2)$  
it follows that
\begin{equation} \label{eq:chebysel}
\begin{split}
& \# \left \{ \frac{x}{(\log x)^A} \le m \le x : \bigg| \sum_{Y \frac{m}{x} < k \le x^{1/3} \frac{m}{x}} \frac{r(m+k)-r(m-k)}{k} \bigg| \ge \frac{(\log x)^{3A}}{\sqrt{Y}} \right\} \\
& \le 9 \frac{Y}{(\log x)^{6A}} 
\sum_{ \frac{x}{(\log x)^A} \le m \le x} \left(\frac{R_m\left(x^{1/2} \frac{m}{x}\right)^2 (\log x)^{2A}}{x}+
 \frac{R_m\left(Y \frac{m}{x}\right)^2 (\log x)^{2A}}{Y^2}+ \left(\int_{Y \frac{m}{x}}^{x^{1/2} \frac{m}{x}} \frac{R_m(t)}{t^2} \, dt \right)^2
\right).
\end{split}
\end{equation}
In the integral we make a change of variables and apply the Cauchy-Schwarz inequality to get for each $m \in [x/(\log x)^A,x]$ that
\begin{equation} \label{eq:cauchyme}
\left(
\int_{Y \frac{m}{x}}^{x^{1/2} \frac{m}{x}} \frac{R_m(t)}{t^2} \, dt \right)^2\le \frac{(\log x)^{2A}}{Y} \int_{Y}^{x^{1/2}} \frac{1}{t^2} R_m\left( t\frac{m}{x}\right)^2 \, dt.
\end{equation}

Observe that
\[
R_m\left(H \frac{m}{x}\right)=\sum_{m \le n \le m+\frac{m}{x} H} r(n)-\sum_{m-\frac{m}{x} H \le n \le m} r(n).
\]
Hence, by Lemma \ref{lem:selberg} with $L=x/H$ (along with an analogue of this lemma for the second sum, which is proved in the same way) we get
\[
\frac{1}{x}\sum_{ m \le x}R_m\left(H \frac{m}{x}\right)^2 \ll H (\log x)^2,
\]
for $1 \le H \le x/3$. Using this bound and \eqref{eq:cauchyme} in \eqref{eq:chebysel} gives
\[
\begin{split}
 &\# \left \{ \frac{x}{(\log x)^A} \le m \le x : \bigg|\sum_{Y \frac{m}{x} < k \le x^{1/2} \frac{m}{x}} \frac{r(m+k)-r(m-k)}{k} \bigg| \ge \frac{(\log x)^{3A}}{\sqrt{Y}} \right\} \\
 &\qquad \qquad \ll \frac{Y \cdot x}{(\log x)^{4A}}\left(\frac{(\log x)^2}{x^{1/2}}+\frac{(\log x)^2}{Y}+\frac{(\log x)^3}{Y} \right) \ll \frac{x}{(\log x)^{4A-3}},
 \end{split}
\]
since we may assume $Y \le x^{1/2}$ otherwise the set on the LHS above is empty.
\end{proof}
Before proving the main result of this section we require the following technical lemma.

\begin{lemma}  \label{lem:trivial}  
Let $u,v$ be sufficiently large positive real numbers such that $v^{9/10} \le u \le 2v$. Let $t > 1$ be a real number, that is not an integer which is expressible as a sum of two squares, such that $|u-t|\le v^{1/3}$. Then
\[
\sum_{m:|m-u| > v^{\frac12}} r(m) \left(\frac{1}{m-t}-\frac{m}{m^2+1} \right)=-\pi \log t+O(1).
\]
\end{lemma}
\begin{proof}
Let $A(x)=\sum_{1 \le n \le x} r(n)=\pi x+E(x)$, it is well-known
that (cf. \cite{sierpinski-circle-problem}) that
$E(x) \ll x^{\frac13}$. Also, let $f_t(x)=\log
\frac{|x-t|}{(x^2+1)^{1/2}}$, (so $f_t(x) \rightarrow 0$ as $x
\rightarrow \infty$). Since $|u-t| \le v^{1/3}$, partial summation
gives 
\[
\begin{split}
\sum_{m:|m-u| > v^{\frac12}} r(m) \left(\frac{1}{m-t}-\frac{m}{m^2+1} \right)=& \int_{u+v^{\frac12}}^{\infty} f_t'(x) dA(x)+\int_{1^-}^{(u-v^{\frac12})^{-}} f_t'(x) dA(x) \\
=& \pi \left(f_t(u-v^{\frac12})-f_t(u+v^{\frac12})-\log t\right)\\
&+O\left(1+\max_{\pm} \frac{u^{\frac13}}{|u\pm v^{\frac12}-t|} \right).
\end{split}
\]
The error is $O(1)$ since we assumed  $|u-t| \le  v^{1/3}$. Also, 
\[
f_t(u-v^{\frac12})-f_t(u+v^{\frac12})=\log \frac{|u-t-v^{\frac12}|}{|u-t+v^{\frac12}|}+O(1) \ll 1.
\]
\end{proof}
We are now ready to prove the main result of this section.

\begin{proof} [Proof of Theorem \ref{thm:spectral}] Let $A \ge 1$.
In the weak coupling quantization, it follows from the spectral equation \eqref{eq:spectral2} along with Lemma \ref{lem:trivial} that
\begin{equation} \label{eq:initial1}
\sum_{ m : |m-n| \le \frac{n}{x} x^{1/2}} \frac{r(m)}{m-\lambda_n}=\pi
\log \lambda_n+O(1) 
\end{equation}
for every integer $\frac{x}{(\log x)^A} \le n \le x$, which is a sum of two squares.
Note that the application of Lemma \ref{lem:trivial} is justified
since it is well-known that $\lambda_n-n \le n^+-n \le 10n^{1/4}$ (see
for instance \cite{MV} p. 43). 

In the strong coupling quantization, applying Lemma \ref{lem:trivial}
twice we get for $\frac{x}{(\log x)^A} \le n \le x$ that 
\[
\bigg|
\sum_{m : |m-n| > \frac{n}{x} x^{1/2}}  r(m)
\left(\frac{1}{m-\lambda_n}-\frac{m}{m^2+1} \right)
-
\sum_{m :|m-\lambda_n| > \lambda_n^{1/2}}  r(m)
\left(\frac{1}{m-\lambda_n}-\frac{m}{m^2+1} \right)\bigg|
\ll 1.
\]
Hence, using this along with the spectral equation \eqref{eq:spectral} we have
 \[
 \begin{split}
 \sum_{|m-n| \le  \frac{n}{x} x^{1/2}}  r(m) \left(\frac{1}{m-\lambda_n}-\frac{m}{m^2+1} \right)=&\sum_{|m-\lambda_n| \le \lambda_n^{1/2}}  r(m) \left(\frac{1}{m-\lambda_n}-\frac{m}{m^2+1} \right)+O(1) \\
 =& \frac{1}{\alpha}+O(1).
 \end{split}
 \]
 Hence, in the strong coupling quantization for each $\frac{x}{(\log x)^A} \le n \le x$
\begin{equation} \label{eq:initial}
\sum_{ m : |m-n| \le \frac{n}{x} x^{1/2}} \frac{r(m)}{m-\lambda_n}
=
\frac{1}{\alpha}+O(1).
\end{equation}

For $\frac{x}{(\log x)^A} \le n \le x$, we now analyze the sum that appears on the LHS of both \eqref{eq:initial1} and \eqref{eq:initial}. Let $B \ge 1$, to be determined later and consider
\begin{equation} \label{eq:pf0}
\sum_{|m-n|\le \frac{n}{x} x^{1/2}} \frac{r(m)}{m-\lambda_n}= \sum_{|m-n|\le  \frac{n}{x}(\log x)^B} \frac{r(m)}{m-\lambda_n}+
\sum_{ \frac{n}{x}(\log x)^B < |k| \le \frac{n}{x} x^{1/2} } \frac{r(n+k)}{k-\delta_n},
\end{equation}
where recall $\delta_n=\lambda_n-n$.
Note that 
\[
\begin{split}
\sum_{\substack{n \le x \\ \delta_n \ge (\log x)^{B/2}}} b(n)
\le & \frac{1}{(\log x)^{B/2}} \sum_{n \le x} b(n)\delta_n \\
\le & \frac{1}{(\log x)^{B/2}} \sum_{n \le x} b(n)  (n^+-n) \ll \frac{x}{(\log x)^{B/2}}.
\end{split}
\]
Hence, for all but $O(x/(\log x)^{B/2})$ integers $n \le x$ which are
representable as a sum of two squares, $\delta_n < (\log
x)^{B/2}$. For 
these integers, with the second sum on the RHS of \eqref{eq:pf0} equals 
\begin{equation} \label{eq:pf11}
\sum_{ \frac{n}{x}(\log x)^B \le k \le \frac{n}{x}x^{1/2} } \frac{r(n+k)-r(n-k)}{k}+O\left( (\log x)^{B/2} \sum_{ \frac{n}{x}(\log x)^B \le |k| \le x^{1/2}} \frac{r(n+k)}{k^2} \right).
\end{equation}
Since
\[
\begin{split}
&\# \bigg\{ \frac{x}{(\log x)^A} \le n \le x : (\log x)^{B/2} \sum_{\frac{n}{x}(\log x)^B \le |k| \le x^{1/2}} \frac{r(n+k)}{k^2} \ge 1 \bigg\}\\
&\le (\log x)^{B/2}  \sum_{(\log x)^{B-A} \le |k| \le x^{1/2}} \frac{1}{k^2}\sum_{ n \le x} r(n+k) \ll  \frac{x}{(\log x)^{B/2-A}}
\end{split}
\]
the $O$-term in \eqref{eq:pf11} is $\ll 1$ for all but $O(x/(\log x)^{B/2-A})$ integers $\frac{x}{(\log x)^A} \le n \le x$. The first
sum in \eqref{eq:pf11} is estimated using Lemma \ref{lem:partial},
with $Y=(\log x)^B$; so for $B \ge 6A$ this sum is $\ll 1$
for all but at most $\ll x/(\log x)^{A}$ integers $n \le x$. Hence,
applying the two previous estimates in \eqref{eq:pf11} and using the resulting
bound along with \eqref{eq:pf0} in \eqref{eq:initial1} and
\eqref{eq:initial} completes the proof upon taking $B \ge 6A$.

% so that by Chebyshev's inequality it follows for all $n=a^2+b^2 \le x$ 

% except outside a set of size $\ll X/(\log X)^A$ that the above equation equals
% \begin{equation} \notag
% \sum_{ \frac{n}{X}(\log X)^B \le |k| \le \frac{n}{X} X^{1/3} } \frac{r(n+k)-r(n-k)}{k}+O\left(\frac{1}{(\log \lambda)^C}\right),
% \end{equation}
% for $B=B(A,C)$ sufficiently large. By Lemma \ref{lem:partial} with $Y=(\log \lambda)^B$  we get that for all $n=a^2+b^2 \le X$ except outside a set of size $O(X/(\log X)^A)$ (here one can sum over dyadic intervals) that
% \[
% \sum_{ \frac{n}{X}(\log X)^B \le |k| \le \frac{n}{X} X^{1/3} } \frac{r(n+k)-r(n-k)}{k} \ll \frac{1}{(\log \lambda)^C}
% \]
% thereby completing the proof.

\end{proof}

\section{Estimates for new eigenvalues nearby almost primes} \label{sec:normal}

In this section we analyze the location of eigenvalues in
$\Lambda_{\text{new}}$ nearby certain integers which are almost
primes. To state the result, let 
\begin{equation} \label{eq:Ndef}
\begin{split}
\mathcal N_1=&\{ n \in \mathbb N :  (1_{ \mathcal P_{\varepsilon}} \ast 1_{ \mathcal P_{\varepsilon}'})(n) \neq 0, \, b(Q_0n+4)=1, \, \& \, \, Q_1 | Q_0n+4\}, \\
\mathcal N_2=&\left\{ n \in \mathcal N_1 : \left(\frac{Q_0n+4}{Q_1}, P(y)\right)=1\right\},
\end{split}
\end{equation}
where $y=x^{\eta}$ with $\eta$ as in Proposition \ref{prop:sieve} and $Q_0,Q_1, \varepsilon, 1_{\mathcal P_{\varepsilon}}$ and  $b(\cdot)$ are as defined in the beginning of Section \ref{sec:sieve}. 
For $j=1,2$ let $\mathcal N_j(x)=\mathcal N_j \cap [1,x]$. In particular,
for each $n \in \mathcal N_2(x)$, $Q_0n+4=Q_1 \mathcal \ell_n$ where $\ell_n$ is an integer which is a sum of two squares. Moreover, since every prime divisor of $\ell_n$ is $\ge y=x^{\eta}$
so for $n \le x$,
$ x^{ \eta \cdot \# \{ p| \ell_n\}} \le \ell_n \le 2Q_0x$
and
\begin{equation}\label{eq:primefactors}
\# \{ p| \ell_n\} \le \frac{2}{\eta}.
\end{equation}

Also, for
a polynomial $R =\sum a_n X^n \in \mathbb Z[X]$, let $\lVert R
\rVert_1=\sum |a_n|$. Note that by Propositions \ref{prop:sieve} and
\ref{prop:sieve2}  
\begin{equation} \label{eq:Nbd}
\begin{split}
\# \mathcal N_1(x) \gg & \varepsilon^2 \frac{1}{\varphi(Q_1)} \frac{x \log \log x}{(\log x)^{3/2}}, \\
\# \mathcal N_2(x) \gg & \varepsilon^2 \frac{Q_0}{\varphi(Q_0 Q_1)} \frac{x \log \log x}{(\log x)^2}.
\end{split}
\end{equation}
Additionally by using an upper bound sieve, it is not difficult to
prove that
\begin{equation} \label{eq:Nbd2}
\begin{split}
\# \mathcal N_1(x) \ll & \varepsilon^2 \frac{1}{\varphi(Q_1)} \frac{x \log \log x}{(\log x)^{3/2}} \\
\# \mathcal N_2(x) \ll & \varepsilon^2 \frac{Q_0}{\varphi(Q_0 Q_1)} \frac{x \log \log x}{(\log x)^2}.
\end{split}
\end{equation}

The main result of this section is the following proposition.
% We also define two sets of integers $\mathcal N_1, \mathcal N_2$ with
% \[
% \mathcal N_1 =\{ n \in \mathbb N : (1_{ \mathcal P_{\varepsilon}} \ast 1_{ \mathcal P_{\varepsilon}'})(n)=1, b(Q_0n+4)=1, Q_1 | Q_0n+4 \}, \qquad \mathcal N_2=\{ n \in \mathcal N_1 : (m_0n+4, P(y))=1\}.
% \]
% By combining our main results from the previous sections, Proposition \ref{prop:sieve} along with Theorem \ref{thm:spectral}, we are able to explicitly describe the location of new eigenvalues which correspond to certain almost primes. 
 
 \begin{proposition} \label{prop:eigenvalue} For all
 $n \in \mathcal N_j(x)$, $j=1,2$, except outside an exceptional set of size 
 \[
 \ll \frac{\# \mathcal N_j(x)}{ \varepsilon^2 (\log \log x)^{1-o(1)}}
 \]
 we have for $m=Q_0n$ that $m^+=m+4$ and
 \[
 \begin{split}
 &\frac{r(m)}{m-\lambda_m}+\frac{r(m^+)}{m^+-\lambda_m} \\
 & \qquad  \qquad =\begin{cases}
 \pi \log \lambda_m+O\left( (\log \log x)^5 \right) & \text{ in the weak coupling quantization},\\
 O\left( (\log \log x)^5 \right) & \text{ in the strong coupling quantization}.
 \end{cases}
 \end{split}
 \]
 \end{proposition}

 We also require a sieve estimate for averages of correlations of
 multiplicative functions. The following result is due to
 Henriot \cite{H}, which
 builds on the work of Nair and Tenenbaum \cite{NT}. See Corollary 1
 of \cite{H} and the subsequent remark therein. Recall that
 $\tau(n)=\sum_{d|n}1$ denotes the divisor function.

\begin{lemma} \label{lem:NT}
Let $R_1(X),\ldots, R_k( X) \in \mathbb Z[ X]$ be irreducible,
pairwise co-prime polynomials, for which each polynomial $R_j$ does
not have a fixed prime divisor. Let $D$ be the discriminant of
$R=R_1\cdots R_k$ and $\varrho_{R_j}(n)=\# \{ a \pmod n: R_j(a)
\equiv 0 \pmod n\}$. Then there exist $C,c_0>0$ such that
for any non-negative multiplicative functions $F_j$, $j=1, \ldots, k$
with $F_j(n) \le \tau(n)$,
we have for $x \ge c_0 \lVert R \rVert_1^{1/10} $ and some $A \ge 1$ that 
\[
\sum_{n \le x} \prod_{j=1}^k F_j(|R_j(n)|) 
\ll  \Delta_D \, x  \prod_{p \le x} \left(1-\frac{\varrho_{R}(p)}{p} \right)  \prod_{j=1}^k \left(\sum_{n \le x}  \frac{F_j(n) \varrho_{R_j}(n)}{n}\right)
\]
where
\[
\Delta_D :=\prod_{p|D}\left(1+\frac{1}{p} \right)^{C},
\]
and the implicit constant, $C$ and $c_0$ depend at most on the degree of $R$.
\end{lemma}
% \begin{remark} \label{rem:trivial}
% For $R_j,F_j$ as above, note that by Hensel's lemma $\varrho_{R_j}(p^k) \ll \tmop{deg R_j}^k$ so we have that
% \[
% \sum_{n \le x} \frac{F_j(n)\varrho_{R_j}(n)}{n} \ll \prod_{p \le x}\left(1+\frac{F_j(p)\varrho_j(p)}{p} \right).
% \]
%  \end{remark}

% \textbf{Remark} The estimate ... follows from a straightforward application of Lemma \ref{lem:NT}. 

We first start with a technical lemma. 
\begin{lemma} \label{lem:correlate}
Let $f$ be a non-negative multiplicative function with $f(n) \le
\tau(n)$ and $f(mn) \le \max\{1,f(n)\} f(m)$ for $m \in \mathbb N$ and
$n$ such that $b(n)=1$. Then for $1 \le |h| \le x^{1/30}$, with $h
\neq 4$ and $j=1,2$, we have
\begin{equation} \label{eq:convolution}
\begin{split}
 \sum_{n \in \mathcal N_j(x)} f(Q_0n+h)   \ll \frac{1}{\varepsilon^2} \cdot g(h)  \prod_{p|Q_0Q_1} \left(1+\frac{1}{p} \right)^C \prod_{p \le x} \left(1+\frac{f(p)-1}{p} \right)  \# \mathcal N_j(x) 
\end{split}
\end{equation}
where $C>0$ is an absolute constant and
\[
g(h)=\tau(|h|)\tau(|h-4|)\prod_{p|h} \left(1+\frac{1}{p} \right)^C\prod_{p|h-4} \left(1+\frac{1}{p} \right)^C.
\]
Additionally (for $h=4$)  there exists $C>0$ such that 
\[
\sum_{n \in \mathcal N_1(x)} f(Q_0n+4) \ll \frac{1}{\varepsilon^2} \cdot  f(Q_1) \prod_{p|Q_0Q_1}\left(1+\frac{1}{p} \right)^C \prod_{\substack{p \le x \\ p \equiv 1 \pamod 4}} \left(1+\frac{f(p)-1}{p} \right)  \# \mathcal N_1(x) .
\]
\end{lemma}
\begin{remark}
When applying this lemma we will take $f(n)=\tfrac14 \cdot r(n),b(n)$ or $2^{-\omega_1(n)}$ where $\omega_1(n)=\# \{ p|n : p \equiv 1 \pmod 4\}$. The hypotheses of the lemma are satisfied for each of these choices.
\end{remark}
\begin{proof}
  Let $T_j=2$ if $j=1$ and $T_j=y$ if $j=2$. Dropping several of the
  conditions on $n \in \mathcal N_j$ we get that (here $q< p$ denote
  primes)
\begin{equation} \label{eq:convolution2}
\begin{split}
 \sum_{n \in \mathcal N_j(x)} f(Q_0n+h) \le 2 \sum_{\substack{q \le \sqrt{x} \\ q\equiv 1 \pamod 4}} \sum_{\substack{p \le x/q \\ Q_1 | Q_0pq+4 \\ (\frac{Q_0pq+4}{Q_1}, P(T_j))=1}} b(Q_0qp+4)f(Q_0qp+h).
    \end{split}
\end{equation}

Let $K=Q_0q$ and $Y=x/q$.  Note that the sum above is empty unless $(K,Q_1)=1$.
Since $(K, Q_1)=1$ there exist integers $\overline{K},\overline Q_1$ with $1 \le | \overline {K}| < Q_1$ and $1 \le |\overline {Q_1}| < K$ such that $K \overline K -Q_1\overline Q_1=1$. Also, for $Z \ge 1$
let $F_Z$ be the totally multiplicative function given by $F_Z(p)=1$
if $p \ge Z$ and zero otherwise. The inner sum on the RHS of
\eqref{eq:convolution2} is bounded by
\begin{equation} \label{eq:correlate2}
\begin{split}
    \ll& \sum_{\substack{n \le Y \\ Q_1|Kn+4}} F_{\sqrt{Y}}(n)F_{T_j}\left(\frac{Kn+4}{Q_1} \right)b(Kn+4)f(Kn+h)+Y^{1/2+o(1)} \\
=&  \sum_{\substack{m \le \frac{Y-4 \overline{K}}{Q_1}}} F_{\sqrt{Y}}(Q_1m-4\overline{K})F_{T_j}\left(Km-4\overline{Q_1} \right)b(KQ_1m-4Q_1\overline{Q_1})f(KQ_1m+h-4K\overline{K})\\
&+O(Y^{1/2+o(1)}) .
\end{split}
\end{equation}
First note $b(KQ_1n-4Q_1\overline{Q_1})=b(Kn-4\overline{Q_1})$. Let $d=(KQ_1,h-4K\overline K)$ and suppose that $h \neq 4$.
We have
\[
f(KQ_1m+h-4K\overline{K}) \le \max\{1, f(d)\} f\left(\frac{KQ_1}{d}m+\frac{h-4K\overline{K}}{d}\right).
\]
Let  $R_1(X)=Q_1  X-4 \overline K$, $R_2(X)=K X-4 \overline Q_1$, $R_3( X)=\frac{KQ_1}{d} X+\frac{h-4K\overline K}{d}$ and $D$ denote the discriminant of $R=R_1R_2R_3$. The polynomials $R_1,R_2,R_3$ and multiplicative functions $F_1=F_{\sqrt{Y}}$, $F_2=F_{T_j}\cdot b$ and $F_3=f$  satisfy the assumptions of Lemma \ref{lem:NT}. Also for $(p,KQ_1)=1$ we have $\varrho_R(p)=3$ and $\varrho_{R_j}(p^k)=1$ for each $j=1,2,3$ and $k \ge 1$, which follows from Hensel's lemma. Hence, the sum in \eqref{eq:correlate2} is bounded by
\[
\begin{split}
\ll & \max\{1,f(d)\} \Delta_{D}   \frac{Y}{Q_1} \prod_{p \le Y} \left(1+\frac{F_{\sqrt{Y
}}(p)+F_{T_j}(p)b(p)+f(p)-3}{p} \right) \prod_{p|KQ_1}\left( 1+\frac{1}{p}\right)^C \\
\ll & \max\{1,f(d)\} \Delta_{D}   \prod_{p|KQ_1}\left(1+\frac{1}{p} \right)^C \frac{Y}{Q_1(\log Y)^{3/2}(\log T_j)^{1/2}} \prod_{p \le Y} \left( 1+\frac{f(p)-1}{p}\right).
\end{split}
\]
Write $d=p_1^{a_1} \cdots p_{\ell}^{a_\ell}$. For each $j=1, \ldots, \ell$ we have $p_j^{a_j}|h$ or $p_j^{a_j}|h-4$ (depending on whether $p_j^{a_j}|K$ or $p_j^{a_j}|Q_1$, respectively); so $f(d) \ll \tau(|h|)\tau(|h-4|)$.
Note the discriminant of $R$ equals $D=16 \frac{K^2Q_1^2}{d^4}h^2(h-4)^2$ so that
\[
\max\{1,f(d)\} \Delta_D \ll g(h) \prod_{p|Q_1K}\left(1+\frac{1}{p} \right)^C.
\]
Also since $Y =x/q\ge \sqrt{x}$, $\prod_{p \le Y} \left( 1+\frac{f(p)-1}{p}\right) \ll \prod_{p \le x} \left( 1+\frac{f(p)-1}{p}\right) $.
Hence, applying the estimates above in \eqref{eq:convolution2}, summing over $q$ and using \eqref{eq:Nbd}
gives the claimed bound for $h \neq 4$. 

For $h =4$ we argue similarly, only now in order to estimate \eqref{eq:correlate2} we use Lemma \ref{lem:NT} with $R_1, R_2$ as before, $R=R_1R_2$ (so the discriminant is $D=16$) and $F_1=F_{\sqrt{Y}}$, $F_2= b \cdot f$. Also noting that here $d=Q_1$  we conclude that \eqref{eq:correlate2} is bounded by
\[
\begin{split}
&\ll f(Q_1) \prod_{p|Q_1K}\left(1+\frac{1}{p} \right)^C \frac{Y}{Q_1 (\log Y)^2} \prod_{p \le x} \left(1+\frac{b(p)f(p)}{p} \right)  \\
&\ll f(Q_1) \prod_{p|Q_1K}\left(1+\frac{1}{p} \right)^C \frac{Y}{Q_1 (\log x)^{3/2}} \prod_{\substack{p \le x \\ p \equiv 1 \pamod 4}} \left(1+\frac{f(p)-1}{p} \right).
\end{split}
\]
Hence, the claim follows in the same way as before.
\end{proof}

\begin{lemma} \label{lem:cheby}
Let $(\log \log x)^4 \le U \le \frac{1}{10}(\log x)^{1/2}$.
There exists $C>0$ such that for all $n \in \mathcal N_j(x)$, $j=1,2$, outside a set of
size
$$\ll
\frac{1}{\varepsilon^2} \cdot  \# \mathcal N_j(x)
\prod_{p|Q_1Q_0}\left( 1+\frac{1}{p}\right)^C \frac{(\log \log
  x)^4}{U}$$ the following hold: 
\begin{equation} \label{eq:space2}
    \sum_{\substack{1 \le |k| \le \frac{1}{U}(\log x)^{1/2} \\ k \neq 4}} b(Q_0n+k)=0, 
\end{equation}
\begin{equation} \label{eq:cor2}
    \sum_{\substack{ 1 \le |k| \le  \frac{n}{x}(\log x)^{B} \\ k \neq 4}} \frac{r(Q_0n+k)}{|k|}  \le U,
\end{equation}
and
\begin{equation} \label{eq:cor3}
    \sum_{ |k| \ge U } \frac{r(Q_0n+k)}{k^2} \le  \frac{1}{\log \log x}.
\end{equation}

\end{lemma}
\begin{proof}
We first establish \eqref{eq:space2}. By Chebyshev's inequality
\begin{equation} \label{eq:chebystep}
    \# \bigg\{ n \in \mathcal N_j(x) :  \sum_{\substack{1 \le |k| \le \frac{1}{U}(\log x)^{1/2} \\ k \neq 4}} b(Q_0n+k) \ge 1 \bigg\}
    \le \sum_{\substack{1 \le |k| \le \frac{1}{U}(\log x)^{1/2} \\ k \neq 4}} \sum_{n \in \mathcal N_j(x)} b(Q_0n+k) .
\end{equation}
Applying Lemma \ref{lem:correlate} to the inner sum  and noting that
\[
\prod_{p \le x} \left(1+\frac{b(p)-1}{p} \right)\ll \frac{1}{\sqrt{\log x}}
\]
we get that the LHS of \eqref{eq:chebystep} is bounded by
\begin{equation}
    \begin{split}
&\ll\prod_{p |Q_1Q_0} \left(1+\frac{1}{p} \right)^C   \frac{\# \mathcal N_j(x)}{\varepsilon^2 \sqrt{\log x}}  \sum_{\substack{1 \le |k| \le \frac{1}{U}(\log x)^{1/2} \\ k \neq 4}}  g(k)\\
&\ll \prod_{p |Q_1Q_0} \left(1+\frac{1}{p} \right)^C \frac{\# \mathcal N_j(x)}{\varepsilon^2} \frac{(\log \log x)^2}{U},
\end{split}
\end{equation}
where the second step follows upon using Lemma \ref{lem:NT}. 

To prove \eqref{eq:cor2}, we argue similarly and apply Lemmas \ref{lem:NT} and \ref{lem:correlate} to get
\begin{equation} \notag
\begin{split}
&\# \bigg\{ n \in \mathcal N_j(x) : \sum_{\substack{ 1 \le |k| \le  \frac{n}{x}(\log x)^{B} \\ k \neq 4}} \frac{r(Q_0n+k)}{|k|}  > U \bigg\} \\
& \qquad \qquad \qquad \le \frac{1}{U} \sum_{\substack{ 1 \le |k| \le  (\log x)^{B} \\ k \neq 4}}  \frac{1}{|k|}\sum_{n \in \mathcal N_j(x)} r(Q_0n+k) \\
& \qquad \qquad \qquad \ll \frac{\# \mathcal N_j(x)}{\varepsilon^2 U} \prod_{p|Q_0Q_1}\left(1+\frac{1}{p} \right)^C \sum_{\substack{ 1 \le |k| \le  (\log x)^{B} \\ k \neq 4}}  \frac{g(k)}{|k|} \\
& \qquad \qquad \qquad  \ll  \frac{\# \mathcal N_j(x)}{\varepsilon^2 U} \prod_{p|Q_0Q_1}\left(1+\frac{1}{p} \right)^C (\log \log x)^3.
    \end{split}
\end{equation}
We will omit the proof of \eqref{eq:cor3} since it follows similarly.
\end{proof}
For almost all $n \in \mathcal N_1(x)$ it is possible to show that $r(Q_0n+4) \asymp (\log n)^{\log 2/2\pm o(1)}$, however since we do not actually need this estimate we will record the weaker estimate below, which suffices for our purposes and is simpler to prove.
\begin{lemma} \label{lem:typical}
Let $\nu>0$ be sufficiently small.
There exists $C>0$ such that for all $n \in \mathcal N_1(x)$ outside a set of size $$\ll \frac{1}{\varepsilon^2} \# \mathcal N_1(x) \frac{(\log \log x)^C}{(\log x)^{\nu}}$$ the following holds
\begin{equation} \label{eq:space}
      (\log x)^{1/4-\nu} \le r(Q_0n+4) \le (\log x)^{1/2+\nu}. 
\end{equation}
\end{lemma}
\begin{proof}
We will only prove the lower bound stated in \eqref{eq:space}.
Let $\omega_1(n)=\sum_{\substack{p |n \\ p \equiv 1 \pmod 4}} 1$. For $n$ which is a sum of two squares $r(n) \ge 2^{\omega_1(n)}$.
 Using this with Chebyshev's inequality and Lemma \ref{lem:correlate} the number of $n \in \mathcal N_1(x)$ which $r(Q_0n+4)<(\log x)^{1/4-\nu}$ is bounded by
\[
\begin{split}
(\log x)^{1/4-\nu}\sum_{n \in \mathcal N_1(x)} 2^{-\omega_1(Q_0n+4)} \ll & (\log x)^{1/4-\nu} \cdot \frac{(\log \log x)^C}{\sqrt{\log x}} \prod_{\substack{p \le x \\ p \equiv 1 \pmod 4}}\left(1+\frac{1}{2p} \right) \cdot \frac{1}{\varepsilon^2} \# \mathcal N_1(x) \\
\ll &  \frac{1}{\varepsilon^2} \# \mathcal N_1(x) \frac{(\log \log x)^C}{(\log x)^{\nu}}
\end{split}
\]
using Lemma \ref{lem:correlate}.
\end{proof}
% Note that since $b(Q_0n+4)=1$ for $n \in \mathcal N_1$, if $r(Q_0n+4) \le (\log x)^{\log 2/2-\delta}$
% then $2^{\omega_1(Q_0n+4)} \le (\log x)^{\log 2/2-\delta}$ so it suffices to bound
% % \[
% % \begin{split}
% % \# \bigg\{ n \in \mathcal N_1(x) : 2^{\omega_1(Q_0n+4)} \le (\log x)^{\log 2/2-\delta} \bigg\}
% % \le (\log x)^{1/4-\delta} \sum_{n \in \mathcal N_1(x) } 2^{-\omega_1(Q_0n+4)} \ll \#\mathcal N_1(x)   (\log x)^{1/4-\delta} \prod_{p \le x} \left( 1+\frac{\frac12\cdot b(p)-1}{p}\right)
% % \]
% The product on the RHS is $\ll (\log x)^{-3/4}$

\begin{proof}[Proof of Proposition \ref{prop:eigenvalue}]
By Theorem \ref{thm:spectral} we get
for all but $O(x/(\log x)^A)$ new eigenvalues $ \lambda_{\ell} \le x$ that
\begin{equation} \notag
\sum_{|m-\ell| \le  \frac{\ell}{x}(\log x)^B} \frac{r(m)}{m-\lambda_{\ell}}=\begin{cases}
\pi \log \lambda_{\ell}+O(1) & \text{in the weak coupling quantization},\\
\frac{1}{\alpha}+O\left(1 \right) & \text{in the strong coupling quantization}.
\end{cases}
\end{equation}
We now consider integers $\ell=Q_0n$ with $n \in \mathcal N_j(x)$,
$j=1,2$ such that the above holds. Using Lemma \ref{lem:cheby}, in
particular \eqref{eq:space2} and \eqref{eq:cor2}  with $U=(\log \log
x)^5$ it follows that for all but $O(\# \mathcal
N_j/(\varepsilon^2(\log \log x)^{1-o(1)}))$ of these integers $n \in
\mathcal N_j(x)$, $j=1,2$, with $\ell=Q_0n$ that $\ell^+=\ell+4$
and 
\[
\begin{split}
\sum_{|m-\ell| \le  \frac{\ell }{x}(\log x)^B} \frac{r(m)}{m-\lambda_{\ell}}=\frac{r(\ell)}{\ell-\lambda_{\ell}}+\frac{r(\ell^+)}{\ell^+-\lambda_{\ell}}+O\left((\log \log x)^5 \right).
\end{split}
\]
Combining the two estimates above completes the proof. \end{proof}

\section{Proofs of the main theorems}

\subsection{Quantization of Observables}
On the unit cotangent bundle $\mathbb S^* M \mathbb \cong 
\mathbb T^2\times S^1$, a smooth function $f \in C^{\infty}(S^1)$ 
has the Fourier expansion
\[
f(x, \phi) =\sum_{ \zeta \in \mathbb Z^2, k \in \mathbb Z} \widehat f(\zeta,k) e^{i \langle x, \zeta \rangle +i k \phi}.
\]
Following Kurlberg and Uebersch\"ar \cite{KU}, we quantize our observables as follows. For $g \in L^2(S^1)$ let
\begin{equation} \label{eq:quantization}
(\tmop{Op}(f) g)(x)=\sum_{ \xi \in \mathbb Z^2 \setminus 0 }\sum_{ \zeta \in \mathbb Z^2 , k \in \mathbb Z} \widehat f(\zeta,k) e^{ik \arg \xi} \widehat g(\xi) e^{i \langle \zeta+\xi, x \rangle}+\sum_{ \zeta \in \mathbb Z^2 , k \in \mathbb Z} \widehat f(\zeta,k) \widehat g(0) e^{i \langle \zeta, x \rangle}.
% +\sum_{ \zeta \in \mathbb Z^2 , k \in \mathbb Z} \widehat a(\zeta, k) \widehat f(0) e^{i \langle \zeta, x\rangle}
\end{equation}
Hence, for pure momentum observables $f : S^1 \rightarrow \mathbb R$ one has
\begin{equation} \label{eq:observe}
(\tmop{Op}(f)g)(x)=\sum_{\xi \in \mathbb Z^2 } f\left( \frac{\xi}{|\xi|} \right) \widehat g(\xi) e^{i \langle \xi, x\rangle}
\end{equation}
and for $\xi=0$, $f(\tfrac{\xi}{|\xi|})$ is defined to be $\int_{S^1} f(\theta) \, \frac{d\theta}{2\pi}$.

Let $g_{\lambda}$ be as given in \eqref{eq:glambda}.
Then for $f$ a pure momentum observable it follows from \eqref{eq:glambda} and \eqref{eq:observe} that
\begin{equation} \label{eq:momentum}
\begin{split}
\langle \tmop{Op}(f) g_{\lambda}, g_{\lambda} \rangle=& \frac{1}{16 \pi^4 }\cdot \frac{1}{\lVert G_{\lambda} \rVert_2^2} \sum_{ n \ge 0 }\frac{1}{(n-\lambda)^2} \sum_{ a^2+b^2=n} f\left(\frac{a+ib}{|a+ib|} \right) \\
=&\frac{1}{\sum_{n \ge 0} \frac{r(n)}{(n-\lambda)^2}}  \sum_{ n \ge 0 }\frac{1}{(n-\lambda)^2} \sum_{ a^2+b^2=n} f\left(\frac{a+ib}{|a+ib|} \right).
\end{split}
\end{equation}

\subsection{Measures associated to sequences of almost primes in narrow sectors}
Let $\mathcal N_1, \mathcal N_2$ be as in \eqref{eq:Ndef}.
Before proceeding to the main result of this section we will specify
our choice of $Q_0,Q_1$. Consider the set of primes
\begin{equation} \label{eq:5prime}
\mathcal S=
\{ p : p=a^2+b^2, 0 \le b \le a \text{ and } \, 0 < \arctan(b/a) \le p^{-1/10} \}
\end{equation}
and let $q_j$ be the $jth$ element of $\mathcal S$. It follows from work of Kubilius \cite{Kub} that
\[
\# \{ p \le x : p \in \mathcal S\} \asymp \frac{x^{9/10}}{\log x},
\]
so $q_j \asymp (j \log j)^{10/9}$. 
Let $ T = \lfloor \log  \log x \rfloor$,   $H= \lfloor 100 \log \log \log x\rfloor$ and
\begin{equation} \label{eq:Qpdef}
Q_0'= \prod_{j=T}^{T+H-1} q_j, \qquad Q_1'=\prod_{j=T+H}^{T+2H-1} q_j.
\end{equation}
% \marginpar{it's possible to let $Q_0,Q_1$ have a different number of
%   factors, but I see no point in this. me neither. UNLESS we can use
%   it to control c.} 
Also, let $r_0, r_1 \in \mathcal S$ with $\tfrac14 \log \log x \le r_0,r_1 \le \frac12 \log \log x$ and $a_0,a_1 \in \mathbb Z$ with $0 \le a_0,a_1 \le \log \log \log x$.
Let  $m_0,m_1$ be integers, which are fixed (in terms of $x$),
whose prime factors are all congruent to  $1 \pmod 4$. 
Write $(m_0,m_1)=p_1^{e_1}\cdots p_s^{e_s}$ and let $g'=\widetilde p_1^{e_1} \cdots \widetilde p_s^{e_s}$
where $\tfrac12 \log \log x < \widetilde p_j < \log \log x$, $\widetilde p_j=c_j^2+d_j^2$ with $0 \le c_j \le d_j$
and $\arctan(c_j/d_j)=\arctan(b_j/a_j)+O(1/(\log \log x)^{1/10})$ where $a_j^2+b_j^2=p_j$ with $0 \le b_j \le a_j$, for each $j=1, \ldots, s$.
We now take 
\begin{equation} \label{eq:Qdef}
Q_0=Q_0'm_0 r_0^{a_0}, \qquad Q_1=Q_1' \frac{m_1}{(m_0,m_1)} r_1^{a_1}g'.
\end{equation}
Note that $(Q_0,Q_1)=1$ and that $Q_0,Q_1 \ll \exp(200 (\log \log \log
x)^2) \le (\log x)^{1/10}$ so that this choice of $Q_0,Q_1$ is
consistent with our prior assumption. 
For $j=1,2$ let
\begin{equation} \label{eq:Mdef}
\mathcal M_j(x)=\{ m \le x : m=Q_0n \text{ and } n \in \mathcal N_j\}. 
\end{equation}
By \eqref{eq:Nbd} and \eqref{eq:Nbd2},
\begin{equation} \label{eq:M1bd}
%    \varepsilon^2 \frac{1}{\varphi(Q_1)} \frac{x \log \log x}{Q_0(\log
%      x)^{3/2}}\ll \mathcal M_1(x) \ll \varepsilon^2
%      \frac{1}{\varphi(Q_1)} \frac{x \log \log x}{Q_0(\log x)^{3/2}}
 \#   \mathcal M_1(x) \asymp \varepsilon^2
   \frac{1}{\varphi(Q_1)} \frac{x \log \log x}{Q_0(\log x)^{3/2}} 
\end{equation}
and
\begin{equation} \label{eq:M2bd}
%   \varepsilon^2 \frac{1}{\varphi(Q_0Q_1)} \frac{x \log \log x}{(\log
%     x)^{2}}\ll \mathcal M_2(x) \ll \varepsilon^2
%     \frac{1}{\varphi(Q_0Q_1)} \frac{x \log \log x}{(\log x)^{2}}.
\#  \mathcal M_2(x) \asymp \varepsilon^2
   \frac{1}{\varphi(Q_0Q_1)} \frac{x \log \log x}{(\log x)^{2}}. 
\end{equation}
We also now assume that
\[
\varepsilon = (\log \log x)^{-1/4}
\]

\begin{lemma} \label{lem:inert}
Let $Q_0,Q_1$ be as in \eqref{eq:Qdef} and $\varepsilon, \eta>0$ be as in Proposition \ref{prop:sieve}. 
Let $m \in \mathcal M_j(x)$, $j=1,2$ where $\mathcal M_j(x)$ is defined as in \eqref{eq:Mdef}. Then for $f\in C^1(S^1)$ with $|f'| \ll1$ 
\begin{equation} \label{eq:inv1}
\frac{1}{r(m)} \sum_{a^2+b^2=m}f\left(\frac{a+ib}{|a+ib|} \right)
=\frac{1}{r(m_0)}  \sum_{a^2+b^2=m_0}f\left(\frac{a+ib}{|a+ib|} \right)+O\left( \varepsilon \right).
\end{equation}
Under the same hypotheses, we have for $m=Q_0n \in \mathcal N_2(x) $
that there exists an integer $\ell_n$ which is a sum of two squares
with $\#\{ p|\ell_n \} \le 2/\eta$ such that 
\begin{equation} \label{eq:inv2}
\frac{1}{r(m^+)} \sum_{a^2+b^2=m^+}f\left(\frac{a+ib}{|a+ib|} \right)
=\frac{1}{r(m_1 \ell_n)}  \sum_{a^2+b^2=m_1\ell_n}f\left(\frac{a+ib}{|a+ib|} \right)+O\left(  \frac{1}{(\log \log x)^{1/11}}\right).
\end{equation}
\end{lemma}
\begin{proof}
First note that for a unit, $u$ of $\mathbb Z[i]$ i.e. $u \in \{\pm 1, \pm i\}$, that for any $n \in \mathbb N$
\begin{equation} \label{eq:unit}
\sum_{a^2+b^2=n}f\left(\frac{u(a+ib)}{|a+ib|} \right)=\sum_{a^2+b^2=n}f\left(\frac{a+ib}{|a+ib|} \right).
\end{equation}

For $m \in \mathcal M_j(x)$ with $j=1$ or $j=2$ write $m=Q_0'm_0r_0^{a_0}n$ where $n \in \mathcal N_j(x)$. 
The factorizations of the ideals $(m)=((a+ib)(a-ib))$ in $\mathbb Z[i]$ are in
one-to-one correspondence with factorizations $(Q_0')=((c+id)(c-id))$,
$(m_0)=((e+if)(e-if))$, $(r_0^{a_0})=((g+ih)(g-ih))$ and $(n)=((k+il)(k-il))$, since $Q_0',m_0,n$ are
pairwise co-prime. Hence, it follows from this and \eqref{eq:unit}
that 
\begin{equation} \label{eq:angle}
 \frac{1}{r(m)}\sum_{a^2+b^2=m} f\left( \frac{a+ib}{|a+ib|}\right) = \frac{1}{r(Q_0')r(m_0)r(r_0^{a_0})r(n)}\sum_{\substack{\alpha \in \mathbb Z[i] \\ \alpha \overline{\alpha}=Q_0'}} \sum_{\substack{ \beta \in \mathbb Z[i] \\ \beta \overline{\beta}=m_0 } } \sum_{\substack{ \gamma \in \mathbb Z[i] \\ \gamma \overline{\gamma}=r_0^{a_0} } } \sum_{\substack{ \delta \in \mathbb Z[i] \\ \delta \overline{\delta}=n } } f\left(\frac{\alpha \beta \gamma \delta}{|\alpha \beta \gamma \delta|} \right).
\end{equation}
Let $\mathcal S$ be as in \eqref{eq:5prime} and write the $j$th
element of $\mathcal S$ as
$q_j=a_j^2+b_j^2$, with $0
\le b_{j} \le a_{j}$. 
By construction,  for $\alpha \in \mathbb Z[i]$ with  $\alpha \overline \alpha=Q_0'$ we can write $\alpha=u \prod_{j \in J} (a_j+\epsilon_jib_j)$ where $J = \{ T, T+1, \ldots, T+H_1-1\}$, $\epsilon_j\in \{ \pm 1\}$ and $u$ is a unit. It follows that
\[
\begin{split}
\frac{\alpha}{|\alpha|}=&u\prod_{j \in J}\frac{a_j+\epsilon_jib_j}{|a_j+ib_j|}\\
=&u\left(1+O\left( \sum_{j \in J } |\arctan(b_j/a_j)| \right)\right)
=u+O\left( \frac{1}{(\log \log x)^{1/11}}\right)
\end{split}
\]
where the unit $u$ depends on $\alpha$.
Also for $\gamma \in \mathbb Z[i]$ with $\gamma \overline \gamma=r_0^{a_0}$, we have
$
\frac{\gamma}{|\gamma|}=u+O(1/(\log \log x)^{1/11})
$ and  for $\delta \in \mathbb Z[i]$ with $\delta \overline \delta=n$, we have
$
\frac{\delta}{|\delta|}=u+O(\varepsilon).
$
Hence by this and \eqref{eq:unit}
\[
\begin{split}
\sum_{\substack{\alpha \in \mathbb Z[i] \\ \alpha \overline{\alpha}=Q_0'}} \sum_{\substack{ \beta \in \mathbb Z[i] \\ \beta \overline{\beta}=m_0 } }  \sum_{\substack{ \gamma \in \mathbb Z[i] \\ \gamma \overline{\gamma}=r_0^{a_0} } } \sum_{\substack{ \delta \in \mathbb Z[i] \\ \delta \overline{\delta}=n } }  f\left(\frac{\alpha \beta \gamma}{|\alpha \beta \gamma|} \right)=&\sum_{\substack{\alpha \in \mathbb Z[i] \\ \alpha \overline{\alpha}=Q_0'}} \sum_{\substack{ \gamma \in \mathbb Z[i] \\ \gamma \overline{\gamma}=r_0^{a_0} } } \sum_{\substack{ \delta \in \mathbb Z[i] \\ \delta \overline{\delta}=n } }  \left(  \sum_{\substack{ \beta \in \mathbb Z[i] \\ \beta \overline{\beta}=m_0 } } f\left(\frac{u_{\alpha, \gamma,\delta} \cdot \beta }{| \beta |} \right) \right)+O(\varepsilon r(m)) \\
=&r(Q_0) r(r_0^{a_0})r(n) \sum_{ a^2+b^2=m_0} f\left(\frac{ a+ib}{| a+ib|}\right)+O(\varepsilon r(m)),
\end{split}
\]
thereby proving \eqref{eq:inv1}.

The proof of \eqref{eq:inv2} follows along the same lines upon noting
that for $m=Q_0n \in \mathcal M_2(x)$ we can write
$m^+=Q_1' r_1^{a_1} \frac{m_1}{(m_1,m_0)} g' \ell_n$ where $\ell_n$ is a sum of two
squares. Note that $Q_1',\frac{m_1}{(m_1,m_0)}, r_1^{a_1}, g',\ell_n$ are pairwise co-prime
by construction since all the prime divisors of $\ell_n$ are $\ge y$;
the latter also implies that $\#\{ p|\ell_n \} \le
2/\eta$.
\end{proof}

\subsection{Proof of Theorem \ref{thm:scar}}
WLOG we can assume all the prime factors of $m_0$ are congruent to $1 \pmod 4$ (see \eqref{eq:angle}).
Let $Q_0,Q_1$ be as in \eqref{eq:Qdef} and $\mathcal M_1(x)$ be as in \eqref{eq:Mdef} and recall for $m \in \mathcal M_1(x)$ that $m=Q_0n$ where $n \in \mathcal N_1(x)$ and $\mathcal N_1$ is as in \eqref{eq:Ndef}. 
By \eqref{eq:space2} and Lemma \ref{lem:typical}
it follows that for all but at most $o(\# \mathcal M_1(x))$ integers
$m \in \mathcal M_1(x)$ that $m^+=m+4$, $(\log x)^{1/4-\nu} \le r(m^+) \le (\log
x)^{1/2+\nu}$ (for any fixed $\nu>0$) and $4 \le r(m) \ll (\log x)^{o(1)}$.
Combining this with Proposition \ref{prop:eigenvalue}  we get that
for all but $o(\# \mathcal M_1(x))$ integers
$m \in \mathcal M_1(x)$ that $\lambda_m-m=o(1)$ and moreover
\begin{equation} \label{eq:eigen1}
\lambda_m-m \asymp 
\begin{cases}
\displaystyle \frac{r(m)}{\log \lambda_m} & \text{ in the weak coupling quantization},\\
\displaystyle \frac{r(m)}{r(m^+)} & \text{ in the strong coupling quantization}.
\end{cases}
\end{equation}
Also note that for such $m$ as above, we also have $|\lambda_m-m^+|
\ge 3$. Hence, using the above estimate along with \eqref{eq:space2}
and 
\eqref{eq:cor3}  with $U= (\log \log x)^5$ we get for all but at most
$o(\# \mathcal M_1(x))$ integers $m
\in \mathcal M_1(x)$ 
  that
(in both  
cases)
\begin{equation} \label{eq:pf1}
\begin{split}
\sum_{\ell \ge 0} \frac{r(\ell)}{(\ell-\lambda_m)^2} =&\frac{r(m)}{(m-\lambda_m)^2}+\frac{r(m^+)}{(m^+-\lambda_m)^2}+o(1) \\
=&\frac{r(m)}{(m-\lambda_m)^2}\left(1+O\left( \frac{r(m^+)(m-\lambda_m)^2}{r(m)} \right) \right)+o(1)\\
=&\frac{r(m)}{(m-\lambda_m)^2}\left(1+o(1) \right).
\end{split}
\end{equation}
Similarly, for all but at most $o(\# \mathcal M_1(x))$ integers $m \in
\mathcal M_1(x)$
\begin{equation} \label{eq:pf2}
\sum_{ \ell \ge 0 }\frac{1}{(\ell-\lambda_m)^2} \sum_{ a^2+b^2=\ell} f\left(\frac{a+ib}{|a+ib|} \right)=
\frac{1}{(m-\lambda_m)^2}\sum_{ a^2+b^2=m} f\left(\frac{a+ib}{|a+ib|} \right)+O(r(m^+)).
\end{equation}
Therefore, combining \eqref{eq:momentum}, \eqref{eq:eigen1},
\eqref{eq:pf1} and \eqref{eq:pf2} it follows for all but at most $o(\#
\mathcal M_1(x))$ integers $m \in \mathcal M_1(x)$ we have
that
\[
\begin{split}
\langle \tmop{Op}(f) g_{\lambda_m}, g_{\lambda_m} \rangle=& (1+o(1)) \frac{(m-\lambda_m)^2}{r(m)} \cdot \left( \frac{1}{(m-\lambda_m)^2}\sum_{ a^2+b^2=m} f\left(\frac{a+ib}{|a+ib|} \right)+O(r(m^+))\right) \\
=&(1+o(1))\frac{1}{r(m)} \sum_{ a^2+b^2=m} f\left(\frac{a+ib}{|a+ib|} \right) +o(1)\\
=& (1+o(1)) \frac{1}{r(m_0)} \sum_{a^2+b^2=m_0} f\left(\frac{a+ib}{|a+ib|} \right)+O(\varepsilon)
\end{split}
\]
where the last step follows by \eqref{eq:inv1}. The estimate for the density of this subsequence of eigenvalues follows immediately from \eqref{eq:M1bd}, noting that $Q_0,Q_1 \ll (\log x)^{o(1)}$.

% Combining \eqref{eq:momentum2} and \eqref{eq:angle} gives that
% \[
% \langle \tmop{Op}(f) g_{\lambda}, g_{\lambda} \rangle=
% \frac{r(p_1p_2)}{r(p_1p_2)r(m_0)} \sum_{a^2+b^2=m_0} f\left( \frac{a+ib}{|a+ib|}\right)+O(\varepsilon) +O\left( \frac{1}{(\log n)^{\theta}}\right)
% \]
% Hence letting $\varepsilon$ tend to zero sufficiently slowly as $n \rightarrow \infty$ we have
% \[
% \lim_{\substack{n \rightarrow \infty \\ n \in \mathcal N}} \langle \tmop{Op}(f) g_{\lambda}, g_{\lambda} \rangle= \frac{1}{r(m_0)} \sum_{a^2+b^2=m_0} f\left( \frac{a+ib}{|a+ib|} \right).
% \]
% This is true for any $m_0$ with $b(m_0)=1$.
% Given $\mu_{\infty} \in \mathcal A$ we now take a sequence $\mathcal M=\{ m : b(m)=1 \}$ such that
% \[
% \frac{1}{r(m)} \sum_{a^2+b^2=m} f\left( \frac{a+ib}{|a+ib|} \right) \rightarrow \int_{S^1} f \, d\mu_{\infty} \qquad (m \rightarrow \infty, m \in \mathcal M)
% \]
% to conclude that
% \[
% \lim_{\substack{ m \rightarrow \infty \\ m \in \mathcal M} } \lim_{\substack{n \rightarrow \infty \\ n \in \mathcal N(m)}} \langle \tmop{Op}(f), |g_{\lambda}|^2 \rangle=\int_{S^1} f \, d\mu_{\infty}.
% \]

\subsection{Proof of Theorem \ref{thm:scar2}} 
WLOG we can assume all the prime factors of $m_0,m_1$ are congruent to $1 \pmod 4$ (see \ref{eq:angle}).
For sake of brevity let $\mathcal L_2=\log \log x$.
Let $Q_0,Q_1$ be as in \eqref{eq:Qdef} and $\mathcal M_2(x)$ be as in
\eqref{eq:Mdef} and recall for $m \in \mathcal M_2(x)$ that $m=Q_0n$
where $n \in \mathcal N_2(x)$ where $\mathcal N_2$ is as in
\eqref{eq:Ndef}. 
Note for each $m \in \mathcal M_2(x)$ that $r(m) \gg \mathcal
L_2^{10}$. Also, by construction $r(m)/r(m+4) \asymp 
\frac{a_0+1}{a_1+1}$ where $H,a_0,a_1$ are also as in
\eqref{eq:Qdef} and note $a_0,a_1 \le \log \mathcal L_2$.  
Applying Proposition \ref{prop:eigenvalue} we get that for all $m \in
\mathcal M_2(x)$ outside an exceptional set of size $o(\#\mathcal
M_2(x))$ that $m^+=m+4$ and 
\begin{equation} \label{eq:prf1}
\frac{\lambda_m-m}{m^+-\lambda_m}=\frac{r(m)}{r(m^+)}\left(1+O\left(\frac{\mathcal L_2^6}{r(m)}\right)\right)=\frac{r(m)}{r(m^+)}\left(1+O\left(\mathcal L_2^{-4}\right)\right).
\end{equation}
In particular, this implies that $\lambda_m-m \gg \mathcal L_2^{-1}$ and $m^+-\lambda_m \gg \mathcal L_2^{-1}$.
As before, using \eqref{eq:space2} and \eqref{eq:cor3}  with $U= \mathcal L_2^5$ we get for all but at most $o(\# \mathcal M_2(x))$ integers $m \in \mathcal M_2(x)$ that 
\begin{equation} \label{eq:prf2}
\begin{split}
\sum_{\ell \ge 0} \frac{r(\ell)}{(\ell-\lambda_m)^2}=\frac{r(m)}{(m-\lambda_m)^2}+\frac{r(m^+)}{(m^+-\lambda_m)^2}+O(\mathcal L_2^{-1})
\end{split}
\end{equation}
and
\begin{equation} \label{eq:prf3}
\begin{split}
\sum_{ \ell \ge 0 }\frac{1}{(\ell-\lambda_m)^2} \sum_{ a^2+b^2=\ell} f\left(\frac{a+ib}{|a+ib|} \right)=&
\frac{1}{(m-\lambda_m)^2}\sum_{ a^2+b^2=m} f\left(\frac{a+ib}{|a+ib|} \right)\\
&+\frac{1}{(m^+-\lambda_n)^2}\sum_{ a^2+b^2=m^+} f\left(\frac{a+ib}{|a+ib|} \right)
+O(\mathcal L_2^{-1}).
\end{split}
\end{equation}
Let $C_m=\frac{1}{1+r(m)/r(m^+)}$. Applying \eqref{eq:prf1},\eqref{eq:prf2}, and \eqref{eq:prf3} in \eqref{eq:momentum} we get
\begin{equation} \label{eq:almost}
\begin{split}
&\langle \tmop{Op}(f) g_{\lambda_m}, g_{\lambda_m} \rangle=(1+O(\mathcal L_2^{-1})) \left(\frac{r(m)}{(m-\lambda_m)^2}+\frac{r(m^+)}{(m^+-\lambda_m)^2} \right)^{-1} \\
& \qquad \times \left( \frac{1}{(m-\lambda_m)^2}\sum_{ a^2+b^2=m} f\left(\frac{a+ib}{|a+ib|} \right)
+\frac{1}{(m^+-\lambda_m)^2}\sum_{ a^2+b^2=m^+} f\left(\frac{a+ib}{|a+ib|} \right)+O(\mathcal L_2^{-1})\right) \\
&\qquad =\frac{C_m}{r(m)} \sum_{ a^2+b^2=m} f\left(\frac{a+ib}{|a+ib|} \right)+\frac{1-C_m}{r(m^+)}\sum_{ a^2+b^2=m^+} f\left(\frac{a+ib}{|a+ib|} \right)+O(\mathcal L_2^{-1}).
\end{split}
\end{equation}
Applying \eqref{eq:inv1} to the first sum above we get
\begin{equation} \label{eq:almost1}
\frac{C_m}{r(m)} \sum_{ a^2+b^2=m} f\left(\frac{a+ib}{|a+ib|} \right)= \frac{C_m}{r(m_0)} \sum_{ a^2+b^2=m_0} f\left(\frac{a+ib}{|a+ib|} \right)+O(\varepsilon).
\end{equation}
Similarly, applying \eqref{eq:inv2} to the second sum on the RHS of \eqref{eq:almost} we get that
\begin{equation} \label{eq:almost2}
\frac{1-C_m}{r(m^+)} \sum_{ a^2+b^2=m^+} f\left(\frac{a+ib}{|a+ib|} \right)= \frac{1-C_m}{r(m_1\ell_n)} \sum_{ a^2+b^2=m_1 \ell_n} f\left(\frac{a+ib}{|a+ib|} \right)+O(\mathcal L_2^{-1/11}),
\end{equation}
for some integer $\ell_n$ with $\#\{p: p|\ell_m\} \le 2/\eta$ by \eqref{eq:primefactors}. Using \eqref{eq:almost1} and \eqref{eq:almost2} in \eqref{eq:almost} completes the proof upon taking $\varepsilon=\mathcal L_2^{-1/2}$. The estimate for the density of this subsequence of eigenvalues follows from \eqref{eq:M2bd}.

\subsection{Proof of Theorem \ref{thm:allmeasures}} \label{sec:twinprimes}
The proof of Theorem \ref{thm:allmeasures} relies on the following hypothesis concerning the distribution of primes.
\begin{hypothesis} \label{hyp:twinprimes}
Let $Q_1,Q_0$ be as in \eqref{eq:Qdef} and $  \varepsilon \ge  (\log \log x)^{-1/2}$ be sufficiently small. Also let $y=x^{\eta}$ where $\eta>0$ is sufficiently small. Then the number of solutions $(u,v) \in \mathbb Z^2$ to 
\[
Q_1u-Q_0v=4
\]
where $v=p_1p_2$ and $u=p_3$ are primes satisfying $1_{\mathcal
  P_{\varepsilon}}(p_1)1_{\mathcal P_{\varepsilon}'}(p_2)1_{\mathcal
  P_{\varepsilon}}(p_3)=1$, $p_3 > y$ such that $ v \le x$
is
\[
\gg \varepsilon^3 \frac{Q_0}{\varphi(Q_0 Q_1)} \frac{x \log \log x}{(\log x)^2}.
\]
where $\mathcal P_{\varepsilon}, \mathcal P_{\varepsilon}'$ are as in \eqref{eq:Pepsidef}.
\end{hypothesis}

% \begin{remark}
% At the cost of a longer proof we can weaken the assumptions in Hypothesis \ref{hyp:twinprimes}. 
% % For example, it would suffice to take $u=p_3p_4$ with $p_4=a^2+b^2$ and $b=o(a)$.
% \end{remark}

% In the above notation, we have that
% \[
% \# \mathcal M_3(x) \gg \varepsilon^3 \frac{1}{\varphi(Q_0Q_1)} \frac{x \log \log x}{(\log x)^2}.
% \]
% \end{hypothesis}
% The above hypothesis may be reinterpreted as a lower bound for the number of solutions $(u,v)$ to the diophantine equation
% \[
% Q_1u-Q_0v=4,
% \]
% over integers $(u,v)$ of the form $u=p_1p_2$ and $v=p_3$ where $p_1,p_2,p_3$ are primes which split in $\mathbb Z[i]$ and have corresponding angle $|\theta_p| \le \varepsilon$. We also assume $p_3 \ge y$ and this last assumption is only for convenience.

\begin{proof}[Proof of Theorem \ref{thm:allmeasures}]
Recall the definition of $\mathcal N_2$ given in \eqref{eq:Ndef}. Let us define
\[
\mathcal N_3=\{ n \in \mathcal N_2 : Q_0n+4=Q_1 p , b(p)=1, \, \&\, |\theta_p| \le \varepsilon\}.
\]
Following \eqref{eq:Mdef} we also define \[\mathcal M_3(x)=\{ m \le x : m=Q_0n \text{ and } n \in \mathcal N_3\}.\]
By Hypothesis \ref{hyp:twinprimes} and  \eqref{eq:M2bd} it follows that 
\begin{equation} \label{eq:M3bd}
\# \mathcal  M_3(x) \asymp \varepsilon \# \mathcal M_2(x)
\end{equation}
where we also have used an upper bound sieve to get that $\# \mathcal M_3(x) \ll \varepsilon \# \mathcal M_2(x)$. Observe that $\mathcal M_3(x) \subset \mathcal M_2(x)$ and the exceptional set in Proposition \ref{prop:eigenvalue} is $o(\# \mathcal M_3(x))$ since we take $\varepsilon =(\log \log x)^{-1/4}$. Hence, we get that \eqref{eq:prf1} holds for $m \in \mathcal M_3(x)$ outside an exceptional set of size $o(\# \mathcal M_3(x))$. Similarly, we can conclude that \eqref{eq:prf2} and \eqref{eq:prf3} also hold for all $m \in \mathcal M_3(x)$ outside an exceptional set of size $o(\# \mathcal M_3(x))$. Therefore, arguing as in \eqref{eq:almost}--\eqref{eq:almost2} we conclude that for $m \in \mathcal M_3(x)$ outside an exceptional set of size $o(\#\mathcal M_3(x))$ we have that
\begin{equation}\label{eq:hyp1}
\langle \tmop{Op}(f) g_{\lambda_m}, g_{\lambda_m} \rangle =\frac{C_m}{r(m_0)} \sum_{ a^2+b^2=m_0} f\left(\frac{a+ib}{|a+ib|} \right)+ \frac{1-C_m}{r(m_1\ell_n)} \sum_{ a^2+b^2=m_1 \ell_n} f\left(\frac{a+ib}{|a+ib|} \right)+O(\mathcal L_2^{-1/11})
\end{equation}
where $m_0,m_1$ are arbitrary, fixed integers whose prime factors are all congruent to $1 \pmod 4$ and $C_m=1/(1+r(m
)/r(m+4))$. By our hypothesis we have that $\ell_n=p$ with $|\theta_p| \le \varepsilon$ and $(m_1,p)=1$. Hence, repeating the argument used to prove \eqref{eq:inv1} it follows that
\begin{equation}\label{eq:hyp2}
\frac{1}{r(m_1\ell_n)} \sum_{ a^2+b^2=m_1 \ell_n} f\left(\frac{a+ib}{|a+ib|} \right)=\frac{1}{r(m_1)} \sum_{ a^2+b^2=m_1 } f\left(\frac{a+ib}{|a+ib|} \right)+O(\varepsilon).
\end{equation}

Given $0 < c <1$ with $c=d/e \in \mathbb Q$
we will now specify our choice of $a_0, a_1$ (from \eqref{eq:Qpdef}). Recall we allow $a_0,a_1$ to grow slowly with $x$ and $Q_0',Q_1'$ have the same number of prime factors. 
Also, by construction $r(\frac{m_1}{(m_0,m_1)} g')=r(m_1)$.
Let $\mathcal L=\lfloor (\log \log \log x)^{1/2} \rfloor$.
We take\[a_0= 2(e-d) r(m_1)\mathcal L \qquad \text{ and } \qquad a_1=dr(m_0) \mathcal L.\]
Hence,
\begin{equation}\label{eq:hyp3}
C_m=\frac{1}{1+\frac{8r(m_0)(a_0+1)}{16 r(m_1)(a_1+1)}}=\frac{d}{e}+o(1).
\end{equation}

We are now ready to complete the proof. Given any attainable measures $\mu_{\infty_0}$, $\mu_{\infty_1}$ and $0 \le c \le 1$ we can take $\{m_{0,j}\}_j$ $\{m_{1,j}\}$ such that $\mu_{0,j}$ weakly converges to $\mu_{\infty_0}$ and  $\mu_{1,j}$ weakly converges to $\mu_{\infty_1}$, as $j \rightarrow \infty$. 
We also take $\{a_{0,j}\}_j,\{a_{1,j}\}_j$ so that $d_j/e_j \rightarrow c$ as $j \rightarrow \infty$. Therefore, by \eqref{eq:hyp1},\eqref{eq:hyp2}, and \eqref{eq:hyp3} we conclude that there exists $\{\lambda_{\ell}\}_{\ell} \subset \Lambda_{\text{new}}$ such that
 \[
 \langle \tmop{Op}(f) g_{\lambda_{\ell}}, g_{\lambda_{\ell}} \rangle \xrightarrow{\ell \rightarrow \infty} c\int_{S^1} f d\mu_{\infty_0} +(1-c) \int_{S^1} f d\mu_{\infty_1}.
 \]
\end{proof}
\appendix
\section{Arithmetic over $\mathbb Q(i)$} \label{app:BV}
% \change{changed section name (large sieve is mentioned in the next
%   subsection heading}

Consider the number field $\mathbb Q(i)$ with ring of integers
$\mathbb Z[i]$.  For $\mathfrak b$ a non-zero integral ideal of
$\mathbb Z[i]$ the residue classes $\alpha \pmod{ \mathfrak b}$,
where $(\alpha)$ and $\mathfrak b$ are relatively prime ideals,
form the multiplicative group $(\mathbb Z[i]/\mathfrak b)^*$.  We now
summarize some well-known facts, which may be found in \cite{Neukirch}
or \cite{IK}.  A \textit{Dirichlet character} $\pmod {\mathfrak b}$ is
a group homomorphism
\[
\chi : (\mathbb Z[i]/\mathfrak b) ^*\rightarrow S^1.
\]
We extend $\chi$ to all of $\mathbb Z[i]$ by setting $\chi(\mathfrak a)=0$ for $\mathfrak a$ and $\mathfrak b$ which are not relatively prime.
% Additionally, consider the unitary character
% \[
% \chi : \mathbb C^* \rightarrow S^1
% \]
% given by
% \[
% \chi(\alpha)=\left(\frac{\alpha}{|\alpha|}\right)^{k}
% \]
% where $k \in \mathbb Z$.  
Let $I$ denote multiplicative group of
non-zero fractional ideals and
$I_{\mathfrak b}=\{\mathfrak a \in I : \mathfrak a \text{ and } \mathfrak b
\text{ are 
  relatively prime}\}$.  A \textit{Hecke Gro\ss encharakter} $\pmod
{\mathfrak b}$ is a homomorphism 
$\psi: I_{\mathfrak b} \rightarrow \mathbb C \setminus \{0\}$ for which
there exists a pair of homomorphisms
\[
\chi \, : \,  (\mathbb Z[i]/\mathfrak b) ^*\rightarrow S^1, \qquad
\chi_{\infty} \, :  \, \mathbb C^* \rightarrow S^1
\]
such that for an ideal $(\alpha)$ with $\alpha \in \mathbb Z[i]$
\[
\psi((\alpha))=\chi(\alpha)\chi_{\infty}(\alpha).
\]
Conversely,
given any $\chi \pmod {\mathfrak b}$ and $\chi_{\infty}$ there exists
a Gro\ss encharakter $\psi \pmod {\mathfrak b}$ such that $\psi=\chi
\cdot \chi_{\infty}$ provided that $\chi( u)\chi_{\infty}( u) = 1$ for
each unit $u \in \mathbb Z[i]$.

In particular, for $4|k$ and $\mathfrak a=(\alpha)$ a non-negative integer 
\[
\psi(\mathfrak a)=\left(\frac{\alpha}{|\alpha|} \right)^{k}
\]
is a Hecke Gro\ss enchakter  $\pmod 1$ and these Hecke Gro\ss
encharakteren can be used to detect primes in sectors.  Additionally,
given a positive rational integer $q$ with $(4,q)=1$ the 
homomorphism
% \change{commented out marginpar (about ex 5, ch 4 in
%   Iwan-Kow)}
% \marginpar{See example 5 at the end of Ch. 4 in
%   Iwan.-Kow., $q$ must be co-prime to the
%   discriminant. }
\[
\chi \circ N : I_q \rightarrow S^1
\]
given by
$(\chi \circ N)(\mathfrak a)=\chi(N(\mathfrak a))$ is a Dirichlet
character $\pmod q$, where $\chi$ is a Dirichlet character $\pmod q$
for $\mathbb Z$, that is $\chi: (\mathbb Z/(q))^* \rightarrow
S^1$, where $N \mathfrak a$ is the norm of $\mathfrak
a$. Hence, for $4|k$  
\[
\psi(\mathfrak a)=(\chi \circ N)( \alpha ) \left( \frac{\alpha}{|\alpha|}\right)^k
\]
is a Hecke Gro\ss encharakter with modulus $q$ and frequency $k$,
where
$\mathfrak a=(\alpha)$. (A priori $\alpha$ is only defined up to 
multiplication by $i$, but for these characters the choice does not matter).
The $L$-function attached to the Gro\ss encharakter $\psi$ given by
\[
L(s, \psi)=\sum_{\mathfrak a} \frac{\psi(\mathfrak a)}{N (\mathfrak a)^s},
\]
has a functional equation and admits an analytic continuation to
$\mathbb C\setminus \{1\}$.
% The analytic conductor, $\mathfrak C$ of $L(s,\psi)$ satisfies $\mathfrak C \ll N(q)(|k|+|s|+3)^2$. {\tt see Duke \cite{duke}. The dependence on $k$ is not clear to me, dependence on $q$ may also be different if $k=0$.
% and for primitive $\psi \pmod q$ satisfies the functional equation
% \[
% L(s, \psi)=L(1-s, \psi) G(s)
% \]
% where
% \[
% G(s)= \pi^{2s-1} q^{1-2s} \frac{ \Gamma(1-s+\tfrac{|k|}{2})}{ \Gamma(s+\tfrac{|k|}{2})}
% \]
% see Fogels, Knapkowski and Duke. Still probably not quite correct....}

Moreover, if $\psi$  is not a real character, $L(s,\psi)$ has a
standard zero free region. That is, we have
\[
L(\sigma+it,\psi) \neq 0 \qquad \text{  for  } \qquad \sigma>1-\frac{c}{ \log (q(|t|+1)(|k|+1))}
\]
(see \cite[Section 5.10]{IK}).
%In particular, unless both $k=0$ and $\chi \pmod q$ is principal
In particular, if $k \ne 0$,
\[
\sum_{N (\pi) \le x } \chi(N(\pi)) \left( \frac{\pi}{|\pi|}\right)^{k} \ll ((|k|+1) q) \cdot x \exp\left( -c \sqrt{\log x}\right),
\]
where the summation is over prime ideal $\mathfrak p=(\pi)$ with norm
$\le x$. 

Furthermore, for $k =0$ the same estimate holds for any complex
$\chi \pmod q$. However for $k=0$ and $\chi \pmod q$
 a real character, there may be a 
possible Siegel zero and in this case we have Siegel's estimate (see
Section 5.9 of 
\cite{IK})
\[
L(\sigma+it , \chi) \neq0 \qquad \text{for } \qquad \sigma \ge 1- \frac{c(\epsilon)}{q^{\epsilon}}
\]
for any $\epsilon>0$. Consequently, we have the Siegel-Walfisz type
prime number theorem for $(a,q)=1$ and $(q,2)=1$ 
% \marginpar{should we
%  give a proof of this? it's more or less standard} 
% I don't think we have to
\begin{equation} \label{eq:PNT}
\sum_{\substack{N(\pi) \le x \\ N(\pi) \equiv a \pamod q \\ 0 \le \arg \pi \le \varepsilon}} 1=\frac{1}{\varphi(q)} \sum_{\substack{N(\pi) \le x \\ (N(\pi),q)=1 \\ 0 \le \arg \pi \le \varepsilon } } 1+O\left(\frac{x}{(\log x)^A} \right)
\end{equation}
for any $A \ge 1$.  (After multiplication by $i^{l}$ for some $l$ we
can ensure that $\theta = \arg i^{l}\pi \in [0,\pi/2)$; we will let
$\arg \pi$  denote this angle.)

%{ \tt There is a bit of a technical issue with even moduli.}
% \change{Commented out ``.. technical issue with even moduli''}

Recall that a prime $p \equiv 3 \pmod 4$ is inert in $\mathbb Z[i]$;
additionally, a prime $p \equiv 1 \pmod 4$ splits in $\mathbb Z[i]$ so
that $p=\pi \overline{\pi}=a^2+b^2$, where $\pi$ is a prime in
$\mathbb Z[i]$. Writing
% \[
% \sum_{\substack{n \equiv a \pamod q }} (1_{\mathcal P}\ast 1_{\mathcal P})(n)=\frac{ \varepsilon}{\varphi(q)} \tmop{li}(x)+O\left(\frac{x}{(\log x)^A} \right)
% \]
% (since here we no longer count $\overline{\pi}$). 
% \marginpar{I'm not sure about this, but it seems that since $N(\pi)=N(\overline{\pi})$ this account for the difference in the two sums. The primes $3 \pmod 4$ with $p \le x$ contribute $O(x^{1/2})$ to the first sum since $N(p)=p^2$. }
\[
\mathcal B(x;q,a,\varepsilon)=\sum_{\substack{n \le x \\ n \equiv a \pamod q}} ( 1_{ \mathcal P_{\varepsilon}} \ast 1_{ \mathcal P_{\varepsilon}'})(n) -\frac{ 1}{\varphi(q)} \sum_{\substack{n \le x \\ (n,q)=1}}( 1_{ \mathcal P_{\varepsilon}} \ast 1_{ \mathcal P_{\varepsilon}'})(n),
\]
formula \eqref{eq:PNT} gives, for $(a,q)=1$ and $(q,2)=1$, that
\begin{equation} \label{eq:SW}
|\mathcal B(x;q,a, \varepsilon)| \ll \frac{x}{(\log x)^A},
\end{equation}
for $q \le (\log x)^A$. In addition it is worth noting that \eqref{eq:PNT} also implies
\begin{equation} \label{eq:asymp}
\sum_{\substack{ n \le x \\ n \equiv a \pamod q}} ( 1_{ \mathcal P_{\varepsilon}} \ast 1_{ \mathcal P_{\varepsilon}'})(n) \sim \frac{4\varepsilon^2}{\varphi(q)} \frac{x \log \log x}{\log x}.
\end{equation}

We are now ready to state the following result which is an analog of the Bombieri-Vinogradov Theorem.
% GRH for all Hecke $L$-functions $L(s,\psi)$ implies that
% \[
% |\mathcal B(x;q,a, \varepsilon)| \ll x^{1/2}(\log x)^{O(1)}
% \]
% for $q \le x^{1/2}$. 
% The main result of this section shows that the estimate above is true on average.
\begin{theorem} \label{thm:bv} 
There exists $B_0$ sufficiently large so that
\[
\sum_{ \substack{q \le Q \\ (q,2)=1}} \max_{(a,q)=1} \left|\mathcal B(x;q,a,\varepsilon) \right| \ll \frac{x}{(\log x)^{10}}
\]
for $Q \le x^{1/2}/(\log x)^{B_0}$.
\end{theorem}
% \marginpar{It would be nice to have a more general theorem than this, we could likely take the max over all sectors with $\varepsilon>1/x^{\theta}$ for some power of $\theta$. This has been essentially done by Coleman and Swallow.}
% Note that this is non-trivial as long as $\varepsilon>1/(\log x)^C$.

Let $\mathcal S \subset \mathbb N$. 
A sequence of complex numbers $\{ \beta_n\}$ with
$|\beta_n| \le \tau(n)$ satisfies the
\textit{Siegel-Walfisz property for $\mathcal S$} provided that for every $q \in \mathcal S$ and $A \ge 0$ and $N \ge 2$ we
have\fixmehide{If we allow general $\mathcal S$ I suppose $B_{0}$ might
  have to depend on the set? \textcolor{blue}{I think $B_0$ will be independent of the set, in the proof the average is split into small moduli and large moduli, the small moduli are handled by S-W the large moduli by the large sieve, and this is where the $B_0$ comes in, but since when using the large sieve we just upper bound the sum over $\mathcal S$ by a sum over $\mathbb N$ $B_0$ should not depend on $\mathcal S$.}}
\[
\sum_{\substack{ n \le N \\ n \equiv a \pmod q}}
\beta_n=\frac{1}{\varphi(q)}\sum_{\substack{n \le N \\ (n,q)=1}}
\beta_n+
%O\left( \frac{N}{(\log N)^A}\right) 
O\left( \frac{N}{(\log N)^A}\right) 
\]
for every $a \in \mathbb Z$ with $(a,q)=1$.

\subsection{An application of the large sieve}
We next recall a consequence of the large sieve, which follows applying
a minor modification of Theorem 9.17 of \cite{FI}.
\begin{lemma} \label{lem:largesieve} 
% \fixme{Perhaps this should be stated for general sets $\mathcal S$ as above? Maybe not}
Let $A \ge 1$ and $Q=x^{1/2}(\log x)^{-B}$ where $B=B(A)$ is sufficiently large. 
Suppose $\{\beta_n\}$ satisfies the Siegel-Walfisz property for all $q$ with $(q,2)=1$.
Then
for any sequence $\{\alpha_n\}$ of complex numbers such that
$|\alpha_n| \le \tau(n)$
\[
\sum_{\substack{q \le Q \\ (q,2)=1}} \max_{(a,q)=1} \Bigg| \sum_{\substack{mn \le x \\m,n \le  \frac{x}{(\log x)^B} \\ mn \equiv a \pamod q}} \beta_m \alpha_n- \frac{1}{\varphi(q)} \sum_{\substack{mn \le x \\m,n \le \frac{x}{(\log x)^B} \\ (mn,q)=1}} \beta_m \alpha_n\Bigg| \ll \frac{x}{(\log x)^A}.
\]
\end{lemma}

\begin{proof}[Proof of Theorem \ref{thm:bv}] By
  \eqref{eq:SW} the 
  sequence $\beta_n=1_{ \mathcal P_{\varepsilon}}(n)$ satisfies the
  Siegel-Walfisz condition for all $q$ with $(q,2)=1$. Take $\alpha_n=1_{ \mathcal
    P_{\varepsilon}'}(n)$ and note that (cf. \eqref{eq:Pepsidef})
    \[
    \sum_{\substack{n \le x \\ n \equiv a \pmod q}} ( 1_{ \mathcal P_{\varepsilon}} \ast 1_{ \mathcal P_{\varepsilon}'})(n)=\sum_{ \substack{mn \le x \\ m,n \le \frac{x}{(\log x)^{B_0}} \\ mn \equiv a \pmod q}}  1_{ \mathcal P_{\varepsilon}}(m)1_{ \mathcal P_{\varepsilon}'}(n)
    \]
    and
    \[
     \sum_{\substack{n \le x \\ (n,q)=1}} ( 1_{ \mathcal P_{\varepsilon}} \ast 1_{ \mathcal P_{\varepsilon}'})(n)=\sum_{ \substack{mn \le x \\ m,n \le \frac{x}{(\log x)^{B_0}}\\ (mn,q)=1}}  1_{ \mathcal P_{\varepsilon}}(m)1_{ \mathcal P_{\varepsilon}'}(n).
    \]
    Hence,
    applying Lemma \ref{lem:largesieve}
  completes the proof. 

\end{proof}

\subsection{Gaussian integers in sectors with norms in
  progressions} 
%   \change{added subsection}
\label{sec:gauss-integ-sect}
The goal of this section is to show that a result of Smith
\cite{smith-circle-problem} (also cf. \cite{tolev})
holds for Gaussian integers in sectors.  We recall that for
$\alpha \in \Z[i]$, $N(\alpha) = |\alpha|^{2}$ denotes the norm of $\alpha$.
For $a,q>0$ define
$$
\eta_{a}(q) :=
|\{ \alpha_1,\alpha_2 \pamod q : \alpha_1^{2}+\alpha_2^{2} \equiv a \pamod q   \}|.
$$

\begin{proposition}
\label{prop:gauss-integ-sect}
Let $a, q>0$ be integers and put  $g = (a,q)$.
  Given an angle $\theta$ and  $\epsilon \in
  (0,2\pi)$, let
  $S=S_{\epsilon,\theta}$ denote the set of lattice points
  $\alpha \in \Z[i]$ contained in the sector defined by\footnote{By
    $\arg(\alpha)$ we 
    denote the complex argument chosen in such a way that it is single
    valued in an $\epsilon/2$-neighborhood of $\theta$.}  
  $ |\arg(\alpha)-\theta| < \epsilon/2 $.   Then, uniformly for
  $\epsilon>0$,
 $$
 |\{ \alpha \in S : N(\alpha) \equiv a \pamod q,
 N(\alpha)\le x \}|
 $$
 $$
 =
 \frac{\epsilon x \eta_{a}(q)}{q^{2}} +O\left( \frac{x^{1-\delta/3}}{q}\right)
 $$
  provided that $q^{3}g < x^{2(1-2\delta)}$ for $\delta>0$. 
 
\end{proposition}

We begin by showing that solutions to $\alpha_{1}^{2} + \alpha_{2}^{2}
\equiv a \pmod q$ is well distributed in fairly small boxes.
Given $q$, let $f : (\Z /q \Z)^{2} \to \C$ denote
the characteristic function of the set $\{ (\alpha_1,\alpha_2) \in (\Z /q \Z)^{2} :
\alpha_1^{2} + \alpha_2^{2} \equiv a \pmod q\}$.
With the modulo $q$ Fourier transform given by
\begin{equation}
  \label{eq:tolev-estimate}
\widehat{f}(\xi_1,\xi_2) :=
\sum_{\alpha_1,\alpha_2 \pamod q} f(\alpha_1,\alpha_2) e^{-2 \pi i (\xi_1 \alpha_1 + \xi_2 \alpha_2)/q}
\end{equation}
we recall the following estimate by Tolev \cite{tolev}:
\begin{equation} \label{eq:tolev}
|\widehat{f}(\xi_1,\xi_2)| \ll
q^{1/2} \tau(q)^{2}
(q,\xi_1,\xi_2)^{1/2} (q,a,\xi_1^{2}+\xi_2^{2})^{1/2}
\leq
q^{1/2} \tau(q)^{2}
(q,\xi_1,\xi_2)^{1/2} (q,a)^{1/2}
\end{equation}

By the Chinese remainder theorem, $\eta_{a}(q)$ is
multiplicative in $q$, and we note that $\widehat{f}(0,0) =
\eta_{a}(q)$. 
%
% For $q$ prime, we record the following special cases:
% $$
% \eta_{a}(q) =
% \begin{cases}
%   1 & \text{if $q \equiv 3 \pmod 4$, $a \equiv 0 \pmod q$,}  \\
%   q+1 & \text{if $q \equiv 3 \pmod 4$, $a \not \equiv 0 \pmod q$,}  \\
%   2q-1 & \text{if $q \equiv 1 \pmod 4$, $a \equiv 0 \pmod q$,}  \\
%   q-1 & \text{if $q \equiv 1 \pmod 4$, $a \not \equiv 0 \pmod q$.}  \\
% \end{cases}
% $$

Let $B \subset [0,q) \times [0,q)$ be a ``box'' with side lengths $T$,
and let $g = g_{B}$ denote the characteristic function of $B \cap
(\Z/q\Z)^2$.  By standard estimates (from summing a geometric series)
we have, for $\xi_1, \xi_2 \neq 0$,
\begin{equation}
  \label{eq:hatg-estimate1}
\widehat{g}(\xi_1,\xi_2) \ll q^{2}/|\xi_1\xi_2|,  
\end{equation}
for $\xi_1 \neq 0$,
\begin{equation}
  \label{eq:hatg-estimate2}
\widehat{g}(\xi_1,0) \ll Tq /|\xi_1|,  
\end{equation}
(and similarly for $\xi_2 \neq 0$), and trivially
$$
\widehat{g}(0,0) = T^{2}.
$$

\begin{lemma}
  \label{lem:modq-box-count}
Let  $g = (a,q)$. Then  
$$
|\{  (\alpha_1,\alpha_2)  \in B : \alpha_1^{2}+\alpha_2^{2} \equiv a \pamod q \}|
=
T^{2} \cdot\frac{\eta_{a}(q)}{q^{2}}
+ O(q^{1/2} \tau(q)^{3} \log(q)^{2}  g^{1/2}  )
$$
\end{lemma}
\begin{proof}
  By Fourier analysis on $(\Z/q\Z)^{2}$ (i.e., Plancherel's theorem
  for finite abelian groups)
  we have
$$
|\{  (\alpha_1,\alpha_2)  \in B : \alpha_1^{2}+\alpha_2^{2} \equiv a \pamod q \}|
=
\sum_{ \alpha_1,\alpha_2 \pamod q } f(\alpha_1,\alpha_2) g(\alpha_1,\alpha_2)
$$
$$
=
\frac{1}{q^{2}}
\sum_{ \xi_1, \xi_2 \pamod q } \widehat{f}(\xi_1,\xi_2) \overline{\widehat{g}(\xi_1,\xi_2)}
$$
The main term is given by $\xi_1=\xi_2=0$ and equals
$$
\frac{\widehat{f}(0,0) \widehat{g}(0,0)}{q^{2}}
= T^{2}\frac{ \eta_a(q)}{q^{2}}
$$

Using \eqref{eq:tolev} and \eqref{eq:hatg-estimate2} the contribution from (say) $\xi_1=0$ and $\xi_2  \neq 0$ is
\begin{equation}
  \label{eq:error-one}
\ll \frac{1}{q^{2}}
\sum_{\xi_2 = 1 }^{q-1}
\frac{Tq}{\xi_2}  q^{1/2} \tau(q)^{2} (q,\xi_2)^{1/2} g^{1/2}
\ll
\frac{Tq^{3/2} \tau(q)^{2} g^{1/2}}{ q^{2}}
\sum_{d|q}
\sum_{0 < \xi_2 < q/d}
\frac{d^{1/2} }{d \xi_2}
\end{equation}
$$
\ll
\frac{T \tau(q)^{3} \log(q) g^{1/2}}{ q^{1/2}} = O(q^{1/2}\tau(q)^{3}
\log(q) g^{1/2} ). \fixmelater{here we pick up an extra factor of $\tau(q)$, which I think we could bound as $\tau(q)^{o(1)}$ if we wanted}
$$
The contribution from terms  $\xi_2=0$ and $\xi_1 \neq 0$ is
bounded similarly.

As for the terms $\xi_1,\xi_2 \neq 0$, we have by \eqref{eq:tolev}
$$
\frac{1}{q^{2}}
\sum_{ \xi_1, \xi_2 \neq 0 } \widehat{f}(\xi_1,\xi_2)
\overline{\widehat{g}(\xi_1,\xi_2)}
\ll
\frac{q^{1/2} \tau(q)^{2}}{q^{2}}
\sum_{ \xi_1, \xi_2 \neq 0 }
\frac{q^{2}}{\xi_1 \xi_2}
(q,\xi_1,\xi_2)^{1/2} g^{1/2}
$$
$$
=
q^{1/2} \tau(q)^{2}
\sum_{d|q}
\sum_{ 0 < \xi_1, \xi_2 \leq q/d }
\frac{d^{1/2} g^{1/2}}{d^{2}\xi_1 \xi_2}
\ll
q^{1/2} \tau(q)^{2} \log(q)^{2} g^{1/2}.
$$
\end{proof}

\begin{proof}[Concluding the proof of
  Proposition~\ref{prop:gauss-integ-sect}] 
  Take $T = x^{(1-\delta)/2}$.  The case $T>q$ is straightforward using a
  simple tiling argument, and we only give details for $T \le q$.

  By a simple
  geometry of numbers argument, we may ``tile'' the sector $S$,
  intersected with a ball of radius $x^{1/2}$, with
  $\epsilon x/T^{2} + O(x^{1/2}/T)$ boxes $B$ (with
  side lengths $T$)
  entirely contained in the sector, and with $O(x^{1/2}/T)$ boxes
  intersecting the boundary.  By Lemma~\ref{lem:modq-box-count}, each
  box $B$ contains 
  $$
T^{2} \cdot\frac{\eta_{a}(q)}{q^{2}}
+ O(q^{1/2} \tau(q)^{2} \log(q)^{2}  g^{1/2}  )
 $$
 points satisfying $ \alpha_1^{2} + \alpha_2^{2} \equiv a \pmod q$.

 As $\eta_{a}(q) < q^{1+o(1)}$  (cf. \cite[Lemma~2.8]{blomer-brudern-dietmann}), we find that the
 number of lattice 
 points in the sector is
\[
 \begin{split} 
& ( \epsilon x/T^{2} + O(x^{1/2}/T))  (
 T^{2} \cdot\frac{\eta_{a}(q)}{q^{2}}
+ O(q^{1/2} \tau(q)^{3} \log(q)^{2}  g^{1/2}  ) )\\
&=\frac{\epsilon \eta_a(q) x}{q^2}+O\left(\frac{x^{1-\delta/2}}{q^{1-o(1)}}+ \epsilon g^{1/2} q^{1/2+o(1)} x^{\delta} \right).
\end{split}
\]
For $q^{3} g < x^{2(1-2\delta)}$ the error term is $\ll \frac{x^{1-\delta/3}}{q} $.
\end{proof}

\subsection{Proof of Lemma \ref{lem:fixed}} 
We may assume $(Q,q)=1$ otherwise the result is trivial.
Let $\delta>0$ be sufficiently small but fixed and set
\[
r_{\varepsilon}(n)=\sum_{\substack{a^2+b^2=n \\ |\arg(a+ib)| \le \varepsilon}} 1.
\]
Also, for $n \in \mathbb N$ and $z>0$ let $\widetilde P_n(z)=\prod_{2 < p <z} p$.
Let $\Lambda_1=\{\lambda_d\}$, $\Lambda'=\{\lambda_e'\}$ be upper bound sieves of level $D=x^{\delta}$ with $(d,2q)=1$ and $(e,2Q)=1$. 
Then for $z=x^{\delta/2}$ we have
\[
\begin{split}
\sum_{ \substack{p=a^2+b^2 \le x \\ |\arg(a+ib)| \le \varepsilon \\ qp+4=Q p_1 \text{where } p_1 \text{ is prime}  }} 1 
\le&  \sum_{m \le qx+4}\sum_{\substack{n \le x \\ qn+4=Qm \\ (m, \widetilde P_q(z))=1 \\ (n,\widetilde P_{Q}(z))=1 }} r_{\varepsilon}(n)+O(x^{\delta/2}) \\
\le& \sum_{m \le qx+4} \sum_{\substack{n \le x \\ qn+4=Qm}} r_{\varepsilon}(n) (\lambda' \ast 1)\left(  n \right) (\lambda \ast 1)\left(m\right)+O(x^{\delta/2}). 
\end{split}
\]
Switching order of summation we have that the sum on the LHS above is
\begin{equation} \label{eq:sieveme}
\begin{split}
=&\sum_{\substack{d,e < D \\(d,e)=1 \\  (d,2q)=1, (e,2Q)=1}} \lambda_d \lambda_e' \sum_{\substack{ n \le x \\ e|n}} r_{\varepsilon}(n) \sum_{\substack{m \le qx+4 \\ d|m \\ qn+4=Qm}} 1 \\
= &\sum_{\substack{d,e < D \\(d,e)=1 \\  (d,2q)=1, (e,2Q)=1}} \lambda_d \lambda_e' \sum_{\substack{n \le x \\ n \equiv \gamma \pamod{Qed}}} r_{\varepsilon}(n) 
\end{split}
\end{equation}
since the inner sum in the first equation above consists of precisely one term provided that $qn+4 \equiv 0 \pmod {Qd}$ and is empty otherwise. Also, here $\gamma=-4e \overline {e} \overline{q}$ where $q\overline q \equiv 1 \pmod{Qd}$ and $e \overline e \equiv 1 \pmod{Qd}$. In particular, $(\gamma, Qed)=e$. 

Let us note some properties of the function $\eta_a(q)$. Recall, $\eta_a(\cdot)$ is multiplicative. Moreover, for $p>2$ and $\ell \ge 1$
\begin{equation} \label{eq:eta1}
\eta_a(p^{\ell})=p^{\ell} \sum_{0 \le j \le \ell} \frac{\chi_4(p)^j}{p^j} c_{p^j}(a)
\end{equation}
and for any $a,q \ge 1$
\begin{equation} \label{eq:eta2}
\eta_q(q) \ll \frac{q^2}{\varphi(q)} \tau((a,q))
\end{equation}
(see \cite[Eqn. (2.20) and Lemma 2.8]{blomer-brudern-dietmann})
where 
\begin{equation} \label{eq:raman}
c_q(a)=\sum_{\substack{b \pamod q \\ (b,q)=1}} e\left(\frac{ab}{q} \right)=\frac{\varphi(q)}{\varphi(q/(q,a))} \mu(q/(q,a))
\end{equation}
is the Ramanujan sum and $\chi_4$ is the non-principal Dirichlet character $\pmod 4$. In particular note that if $(a,q)=g$ then $\eta_a(q)=\eta_g(q)$ for odd $q$. 

By Proposition \ref{prop:gauss-integ-sect}, \eqref{eq:eta1},
\eqref{eq:eta2} and recalling that 
$(Qed, \gamma)=e$ we get the RHS of \eqref{eq:sieveme} equals
\[
\begin{split}
&2\varepsilon x \sum_{\substack{d,e < D \\(d,e)=1 \\  (d,2q)=1, (e,2Q)=1}} \frac{\lambda_d \lambda_e'}{(Qed)^2} \eta_{\gamma}(Qed)+O\left(\frac{ x^{1-\delta/4}}{Q}\right)\\
&=\frac{2\varepsilon x \eta_1(Q)}{Q^2}\sum_{\substack{d,e < D \\(d,e)=1 \\  (d,2q)=1, (e,2Q)=1}} \frac{\lambda_d \lambda_e'}{(ed)^2} \frac{\eta_{1}(Qd) \eta_{e}(e)}{\eta_1(Q)}+O\left( \frac{ x^{1-\delta/4}}{Q}\right)
\end{split}
\]
provided that $ Q^3 D^7< x^{2(1-2\delta)}$ which we rewrite as $Q <  x^{2/3-11 \delta/3}$.
Using Theorem \ref{thm:conv} in the form of \eqref{eq:sievebd}, and noting that $\eta_1(Qd)/\eta_1(Q)$ is a multiplicative function, we get that the above sum is
\begin{equation} \label{eq:euler}
\begin{split}
&\ll  \frac{\varepsilon x \eta_1(Q)}{Q^2}    \prod_{\substack{p < D \\ (p,2q)=1}}\left(1-\frac{\eta_1(Qp)}{p^2 \eta_1(Q)}\right) \prod_{\substack{p < D \\ (p,2Q)=1}}\left(1-\frac{\eta_p(p)}{p^2}\right).
\end{split}
\end{equation}
To evaluate the Euler products we use \eqref{eq:eta1} to get $\eta_p(p)=p(1+\chi_4(p)-\frac{1}{p})$, $\eta_1(Qp)/\eta_1(Q)=p+O(1)$ and $\eta_1(Q)=Q\prod_{p|Q}\left(1-\frac{\chi_4(p)}{p} \right)$. Hence, by these estimates we get that \eqref{eq:euler} is
\[
\begin{split}
\ll & \frac{\varepsilon x \eta_1(Q)}{Q^2 }  \prod_{p|Q} \left( 1+\frac{\chi_4(p)+1}{p} \right) \prod_{p|q} \left(1+\frac{1}{p} \right) \cdot \frac{1}{(\log D)^2} 
\\
\ll & \frac{q}{\varphi(q)} \cdot \frac{\varepsilon x }{Q\delta^2 (\log x)^2 }  \prod_{p|Q}\left(1+\frac{1}{p} \right) \ll \frac{q}{\varphi(q) } \cdot \frac{\varepsilon x }{\varphi(Q) \delta^2 (\log x)^2 }
\end{split}
\]
for $Q < x^{2/3-11\delta/3}$
which completes the proof, since $\delta>0$ is arbitrary.
% \fixme{the only limitation on $q$ seems to be $q \le x^{A}$ for some $A \ge 1$}

%\section{}

\section{Non-attainable quantum limits}
\label{sec:non-atta-quant}
Given an integer $n$ such that $r(n)>0$, define
a probability measure $\mu_n$ on the unit circle by
$$
\mu_{n} := \frac{1}{r(n)}\sum_{\lambda \in \Z[i]: |\lambda|^{2}=n}
\delta_{\lambda/|\lambda|},
$$
i.e., $\mu_n$ is obtained by projection the set of $\Z^{2}$-lattice
points on a circle of radius $n^{1/2}$ to the unit circle and $\delta$ here denotes the Dirac delta function.  A measure
$\mu$ is said to be {\em attainable} if $\mu$ is a weak* limit of some
subsequence of measures $\mu_{n_{i}}$.
A partial classification of the set of attainable measures were given
in \cite{KW} in terms of their Fourier coefficients.  Namely, 
for $k \in \Z$, let $\widehat{\mu}(k) := \int z^{k} \, d\mu(z)$
denote the $k$-th Fourier coefficent of $\mu$.  By
\cite[Theorem~1.3]{KW}, the inequalities 
$$
2 \widehat{\mu}(4)^{2}-1 \le \widehat{\mu}(8) \le
\max( \widehat{\mu}(4)^{4}, (2|\widehat{\mu}(4)|-1)^{2})
$$
holds if $\mu$ is attainable.  In particular, for $\gamma>0$ small
and $\widehat{\mu}(4) = 1-\gamma$, we must have $\widehat{\mu}(8)  =
1-4\gamma + O(\gamma^{2})$.  

Now, by Theorem~\ref{thm:scar2}, there exists quantum limits that are
convex combinations $c \nu_{1} + (1-c) \nu_{2}$ for $c>0$ arbitrary
small, and where $\nu_1$ is the uniform measure (with
$(\widehat{\nu_{2}}(4), \widehat{\nu_2}(8)) = (0,0)$), 
and $\nu_{2}$ is a Cilleruello type measure, i.e., localized on the
four points $\pm 1, \pm i$, and with 
$(\widehat{\nu_{2}}(4), \widehat{\nu_2}(8)) =
(1,1)$.  Clearly 
such convex combinations cannot be attainable for $c$ small.

\bibliographystyle{abbrv} 
%\bibliography{paper,mybib}
\bibliography{paper}

\end{document}